\documentclass[a4paper,12pt]{article}
\usepackage[utf8]{inputenc}
\usepackage{amsmath}  
\usepackage{amssymb}       
\usepackage{amsfonts}
\usepackage{dsfont}
\usepackage{times}
\usepackage{amsthm} 
\usepackage{epsfig}
\usepackage{paralist}         
\usepackage{color}
\newcommand{\DATUM}{14-Apr-2020}       
\newcommand{\ol}{\overline}   
\newcommand{\ul}{\underline}  

\newcommand{\eps}{{\varepsilon}}        
\newcommand{\Om}{\Omega}                

\newcommand{\la}{\langle}
\newcommand{\ra}{\rangle}

\newcommand{\dGamma}{{\mathrm{d}\Gamma}}
\newcommand{\qB}{B}

\newcommand{\Bog}{{\mathrm{Bog}}}

\newcommand{\Vol}{{\mathrm{Vol}}}
\newcommand{\QF}{{\mathrm{qfDM}}}
\newcommand{\bfone}{\mathbf{1}}

\newcommand{\cB}{\mathcal{B}}
\newcommand{\cD}{\mathcal{D}}
\newcommand{\cE}{\mathcal{E}}
\newcommand{\cF}{\mathcal{F}}
\newcommand{\cG}{\mathcal{G}}
\newcommand{\cH}{\mathcal{H}}

\newcommand{\cJ}{\mathcal{J}}

\newcommand{\cL}{\mathcal{L}}         
\newcommand{\cO}{\mathcal{O}}         

\newcommand{\cS}{\mathcal{S}}

\newcommand{\field}[1]{\mathbb{#1}}
\newcommand{\bbA}{\field{A}} 
\newcommand{\vbbA}{\vec{\field{A}}} 
\newcommand{\bbU}{\field{U}} 
\newcommand{\HH}{\field{H}} 
\newcommand{\tHH}{\field{H}} 
\newcommand{\RR}{\field{R}}     
\newcommand{\ZZ}{\field{Z}}     
\newcommand{\NN}{\field{N}}     
\newcommand{\CC}{\field{C}}     
\newcommand{\bbW}{\field{W}}     
\newcommand{\fF}{\mathfrak{F}}
\newcommand{\fh}{\mathfrak{h}}  

\newcommand{\sfh}{\mathsf{h}}
\newcommand{\sfj}{\mathsf{j}}
\newcommand{\sfJ}{\mathsf{J}}
\newcommand{\sfN}{\mathsf{N}}  
\newcommand{\sfR}{\mathsf{R}}  
\newcommand{\sfq}{\mathsf{q}}  
 
\newcommand{\sfu}{\mathsf{u}} 
\newcommand{\sfv}{\mathsf{v}}
\newcommand{\bchi}{{\overline{\chi}}}

\newcommand{\hcE}{\widehat{\mathcal{E}}}

\newcommand{\hr}{\hat{r}}

\newcommand{\hphi}{\hat{\phi}}

\newcommand{\tT}{\widetilde{T}}

\newcommand{\ta}{\tilde{a}}
\newcommand{\tb}{\tilde{b}}
\newcommand{\td}{\tilde{d}}
\newcommand{\tg}{\tilde{g}}   

\newcommand{\vA}{{\vec{A}}}             

\newcommand{\vnabla}{{\vec{\nabla}}}
\newcommand{\veps}{{\vec{\epsilon}}}

\newcommand{\vk}{{\vec{k}}}
\newcommand{\vm}{{\vec{m}}}

\newcommand{\vp}{{\vec{p}}}
\newcommand{\vv}{{\vec{v}}}
\newcommand{\vx}{{\vec{x}}}

\newcommand{\DM}{\mathfrak{DM}}

\newcommand{\rIm}{\mathrm{Im}}
\newcommand{\rRe}{\mathrm{Re}}               

\newcommand{\cirS}{\mathop{\bigcirc\kern -.73em {\scriptstyle{\rm S}}}}

\newcommand{\dom}{\mathrm{dom}} 
\newcommand{\supp}{\mathrm{supp}}

\newcommand{\Tr}{{\rm Tr}}

\newcommand{\op}{\mathrm{op}} 
\newcommand{\gs}{{\rm gs}}

\newcommand{\hf}{H_{\mathrm{ph}}} 

\newcommand{\cre}{A_J^*}
\newcommand{\ann}{A_J}
\newcommand{\cresf}{A_{\sfJ}^*}
\newcommand{\annsf}{A_{\sfJ}}

\newcommand{\LL}{{\mathrm{LL}}}
\newcommand{\el}{{\mathrm{el}}}
\newcommand{\ph}{{\mathrm{ph}}}
\newcommand{\pol}{{\mathrm{pol}}}
\newcommand{\rmH}{\mathrm{H}}

\numberwithin{equation}{section}

\newtheorem{theorem}{Theorem}[section]        
\newtheorem{lemma}[theorem]{Lemma}             
\newtheorem{corollary}[theorem]{Corollary}     
\theoremstyle{plain}
\begin{document}
\bibliographystyle{plain}
\title{On the Ultraviolet Limit \\ 
of the Pauli-Fierz Hamiltonian \\
in the Lieb-Loss Model}

\author{Volker Bach \thanks{Institut f\"ur Analysis und Algebra,
    Technische Universit\"at Braunschweig, Germany,
    $<$v.bach@tu-bs.de$>$, ORCID 0000-0003-3987-8155} 
  \and
  Alexander~Hach \thanks{Institut f\"ur Analysis und Algebra,
    Technische Universit\"at Braunschweig, Germany,
    $<$a.hach@tu-bs.de$>$}
}

\date{\DATUM}

\maketitle

\begin{abstract}
  \noindent Two decades ago, Lieb and Loss \cite{LiebLoss1999}
  proposed to approximate the ground state energy of a free,
  nonrelativistic electron coupled to the quantized radiation field by
  the infimum $E_{\alpha, \Lambda}$ of all expectation values $\langle
  \phi_{el} \otimes \psi_{ph} | H_{\alpha, \Lambda} (\phi_{el} \otimes
  \psi_{ph}) \rangle$, where $H_{\alpha, \Lambda}$ is the
  corresponding Hamiltonian with fine structure constant $\alpha >0$
  and ultraviolet cutoff $\Lambda < \infty$, and $\phi_{el}$ and
  $\psi_{ph}$ are normalized electron and photon wave functions,
  respectively. Lieb and Loss showed that $c \alpha^{1/2}
  \Lambda^{3/2} \leq E_{\alpha, \Lambda} \leq c^{-1} \alpha^{2/7}
  \Lambda^{12/7}$ for some constant $c >0$. In the present paper we
  prove the existence of a constant $C < \infty$, such that
  \begin{align*}
  \bigg| \frac{E_{\alpha, \Lambda}}{F[1] \, \alpha^{2/7} \, \Lambda^{12/7}} 
  - 1 \bigg|
  \ \leq \ 
  C \, \alpha^{4/105} \, \Lambda^{-4/105} 
  \end{align*}
  holds true, where $F[1] >0$ is an explicit universal number. This
  result shows that Lieb and Loss' upper bound is actually sharp and
  gives the asymptotics of $E_{\alpha, \Lambda}$ uniformly in the
  limit $\alpha \to 0$ and in the ultraviolet limit $\Lambda \to
  \infty$.
\end{abstract}


\thispagestyle{empty}

\newpage
\setcounter{page}{1}

\section{Introduction and Result} \label{sec-I}
%
\setcounter{equation}{0}
Soon after the discovery of quantum mechanics almost a century ago by
Heisenberg and Schr{\"o}dinger, the quantization of the radiation
field was formulated by Born, Heisenberg, and Jordan and by Dirac
\cite{BornHeisenbergJordan1926, Dirac1927}, and about seventy years ago
quantum electrodynamics (QED) was formulated by Feynman, Schwinger,
Tomonaga, and Dyson \cite{Dyson1949, Feynman1949, Schwinger1948,
  Tomonaga1946}, laying the foundation to answer the question whether
light rays consisted of particles or waves that was open for several
centuries. Besides being conceptually satisfying, QED is one of the
most successful theories with quantitative predictions that match
experimental data by more that eight decimals.

In spite of its success for applications, however, QED is still
lacking essential parts of its mathematical foundation to this very
day. Namely, all known formulations require unphysical regularizations
at large, ultraviolet, and/or small, infrared, photon energies. While
considerable progress has been made in the past three decades on the
construction of the infrared limit, i.e., the construction of a theory
without regularization at small photon energies
\cite{BachFroehlichSigal1998a, BachFroehlichSigal1998b,
  GriesemerLiebLoss2001, BarbarouxChenVugalter2003,
  BachFroehlichPizzo2006}, the construction of the ultraviolet limit
is wide open. This difficult problem has been tackled from several
angles, e.g., by replacing the fully interacting model by effective
mean-field theories of various kinds \cite{HainzlLewinSolovej2007,
  GravejatLewinSere2009, GravejatLewinSere2018}.

One approach among these is a simplifying variational model proposed
by Lieb and Loss in 1999 \cite{LiebLoss1999}. Their starting point is the
Pauli-Fierz Hamiltonian
\begin{align} \label{eq-0.01}
H_{\alpha, \Lambda} 
\ = \ 
\frac{1}{2} \Big( \tfrac{1}{i} \vnabla_x 
- \alpha^{1/2} \vA_\Lambda(x) \Big)^2 
+ H_\ph
\end{align}
of a nonrelativistic spinless particle (modelling the electron),
minimally coupled to the quantized radiation field. Here
$\tfrac{1}{i} \vnabla_x$ is the (particle) momentum operator and
$\vA_\Lambda(x) = \int_{|k| \leq \Lambda} \big( e^{-ik \cdot x} a^*(k) +
e^{ik \cdot x} a(k) \big) \frac{\eps(k) \, dk}{(2\pi)^{3/2} \, |k|^{1/2}}$ is
the magnetic vector potential in Coulomb gauge and cut off for momenta
larger than $\Lambda$ in magnitude. Moreover, $H_\ph = \int |k| \;
a^*(k) \, a(k) \; dk$ is the energy of the radiation field, and
$\alpha \approx 1/137$ is the (dimensionless) fine structure constant.
The Hamiltonian $H_{\alpha, \Lambda}$ is an unbounded, self-adjoint
operator on the domain $\dom[H_{0,0}] \subseteq \cH_\el \otimes
\cF_\ph$ of the noninteracting Hamiltonian $H_{0,0} = \frac{1}{2}
(-\Delta) \otimes \bfone_\ph + \bfone_\el \otimes H_\ph$, see
\cite{Hiroshima2000b, HaslerHerbst2008}, where $\cH_\el = L^2(\RR^3)$
is the space of square-integrable functions on $\RR^3$, and $\cF_\ph$
is the boson Fock space over the space $L^2(\RR^3 \times \ZZ_2)$ of
square-integrable, purely transversal vector fields, see
Section~\ref{sec-II} for a precise definition.

Note that $H_{\alpha, \Lambda} \geq 0$ as a quadratic form. The
(nonnegative) ground state of the energy of the system is
characterized by the Rayleigh-Ritz variational principle as the
infimum of all energy expectation values of the system,
\begin{align} \label{eq-0.02}
E_\gs(\alpha, \Lambda)
\ := \ 
\inf\Big\{ \la \Psi \, | H_{\alpha, \Lambda} \Psi \ra \; \Big|
\ \Psi \in \cH_\el \otimes \cF_\ph \, , \ \ 
\|\Psi\| = 1 \; \Big\} \, .
\end{align}
Lieb and Loss proposed \cite{LiebLoss1999} to restrict the variation in
\eqref{eq-0.02} to wave functions of product form 
$\Psi = \phi \otimes \psi$, with normalized $\phi \in \cH_\el$ and
$\psi \in \cF_\ph$, to obtain a new approximation and upper bound
$E_\LL(\alpha, \Lambda) \geq E_\gs(\alpha, \Lambda)$ to the ground 
state energy, i.e.,
\begin{align} \label{eq-0.03}
E_\LL(\alpha, \Lambda)
\ := \
\inf\Big\{ \cE_{\alpha, \Lambda}(\phi, \psi)  \; \Big|
\ & \phi \in \cH_\el \, , \ \psi \in \cF_\ph \, , \ 
\|\phi\| = \|\psi\| = 1 \; \Big\} \, ,
\\[1ex] \label{eq-0.04}
\cE_{\alpha, \Lambda}(\phi, \psi)  
\ := \ &
\big\la \phi \otimes \psi \big|
H_{\alpha, \Lambda} (\phi \otimes \psi) \big\ra \, .
\end{align}
Note that \textit{upper bounds} on the ground state energy are of
particular interest here because the ultraviolet problem is about the
understanding of the divergence of $E_\gs(\alpha, \Lambda) \to
\infty$, as $\Lambda \to \infty$. We henceforth refer to
Eqs.~\eqref{eq-0.03}-\eqref{eq-0.04} as the \textbf{Lieb-Loss Model}.

In Theorem~1.1 in \cite{LiebLoss1999} Lieb and Loss proved the
existence of two universal constants $C_1, C_2 \in \RR^+$ such that
\begin{align} \label{eq-0.05}
C_1 \, \alpha^{1/2} \, \Lambda^{3/2} 
\ \leq \ 
E_\LL(\alpha, \Lambda)
\ \leq \ 
C_2 \, \alpha^{2/7} \, \Lambda^{12/7} \, . 
\end{align}
This is the first of a series of results of Lieb and Loss in
\cite{LiebLoss1999}, extending their model to $N \geq 2$ fermions or
bosons, taking the electron spin into account by studying the Pauli
operator, and replacing the nonrelativistic kinetic energy by a
pseudorelativistic one. Note that the Lieb-Loss model does not take
the renormalization of the electron mass into account, and the actual
value of $E_\LL(\alpha, \Lambda)$ is of limited quantitative use in
physics. The significance of Eq.~\eqref{eq-0.05}, however, lies in the
fact that the formal perturbation expansion of the ground state about
the photon vacuum yields $E_\gs(\alpha, \Lambda) \sim C \alpha
\Lambda^2$. In contrast, Eq.~\eqref{eq-0.05} says that this grossly
overestimates the ground state energy; it is a warning sign that
perturbation theory may not be adequate to construct the ultraviolet
limit.

The main result of this paper is to derive the asymptotics of
$E_\LL(\alpha, \Lambda)$, as $\Lambda \to \infty$ or $\alpha \to
0$. We obtain an exact characterization of the ground state and ground
state energy of the Lieb-Loss Model for any given $\alpha >0$ and
$\Lambda \geq 1$, in terms of an auxiliary classical functional. To
formulate this precisely, we introduce
\begin{align} \label{eq-0.06}
\cF_\beta(\phi)  
\ := \ &
\frac{1}{2} \big\| \vnabla \phi \big\|_2^2 
+ \beta \, \|\phi\|_1 \, ,
\end{align}
for all $\phi \in Y := H^1(\RR^3) \cap L^1(\RR^3)$, where 
$\|f\|_p := (\int |f(x)|^p \, d^3x)^{1/p}$ denotes the usual
$L^p$-norm, here and henceforth. It is not hard to see
that
\begin{align} \label{eq-0.07}
F[\beta] \ := \ &
\inf\big\{ \cF_\beta(\phi) \: \big| \ 
\phi \in Y \, , \ \ \|\phi\|_2 = 1 \big\} 
\end{align}
satisfies the scaling relation
\begin{align} \label{eq-0.07a}
F[\beta] \ = \ \beta^{4/7} \, F[1] \, , 
\end{align}
and in \cite{Hach2020} the second author shows that the infimum in
\eqref{eq-0.07} is actually attained and strictly positive, in
particular,
\begin{align} \label{eq-0.07b}
F[1] \ > \ 0 \, .
\end{align}
Our main result is estimate~\eqref{eq-0.08} below, showing that the
upper bound on \break $E_\LL(\alpha, \Lambda)$ in \eqref{eq-0.05} is actually
tight.
\begin{theorem} \label{thm-0.1} There exists a universal
constant $C < \infty$ such that for all $\alpha>0$ and 
$\Lambda \geq 1$, the estimate
\begin{align} \label{eq-0.08}
- C \, \alpha^{\frac{4}{49}} \Lambda^{-\frac{4}{49}}
\ \leq \ 
\frac{E_\LL(\alpha, \Lambda)}{F_1 \, \alpha^{2/7} \, \Lambda^{12/7}} 
\: - \: 1 
\ \leq \  
C \, \alpha^{\frac{4}{105}} \Lambda^{-\frac{4}{105}} 
\end{align}
holds true.
\end{theorem}
We briefly sketch the derivation of \eqref{eq-0.08}. The intermediate
steps yield further insight on the minimizer of the Lieb-Loss model.
The latter is described in detail in Section~\ref{subsec-III.4}.
\begin{itemize}
\item[(1)] For technical reasons we introduce an infrared cutoff
  $\sigma > 0$. The case $\sigma = 0$ can be dealt with by a
  continuity argument in the limit $\sigma \to 0$ using standard
  relative bounds on $\vA_\sigma$. We do not give details of
  the argument but refer the reader to \cite{BachFroehlichSigal1998a}.

\item[(2)] We first analyze the functional $\cE_{\alpha, \Lambda}$. A direct
computation yields
\begin{align} \label{eq-0.09}
\cE_{\alpha, \Lambda}(\phi, \psi)  
\ = \ 
\frac{1}{2} \big\| \vnabla \phi \big\|_2^2
+ \Big\la \psi \Big| \;
\HH\big( |\phi|^2 , \, \rIm\{ \ol{\phi} \, \vnabla\phi \} \big) 
\psi \Big\ra_{\cF} \, ,
\end{align}
where $\la \cdot | \cdot \ra_\cF$ denotes the scalar product on the
photon Fock space $\cF_\ph$ and $\HH[\rho, \vv]$ is for 
$\rho: \RR^3 \to \RR^+$ and $\vv: \RR^3 \to \RR^3$ given as
\begin{align} \label{eq-0.10}
\HH[\rho, \vv]
\ := \ 
\hf + \frac{\alpha}{2} \int \rho(x) \, \vA_{\sigma,\Lambda}^2(x) \, d^3x +
\sqrt{\alpha}  \int \vv(x) \cdot \vA_{\sigma,\Lambda}(x) \, d^3x \, .
\end{align}
In Theorem~\ref{thm-IV-vb.04} in Section~\ref{subsec-IV.1} we
demonstrate that, by a suitably chosen Weyl transformation $W_\phi$,
the term linear in the fields, i.e., proportional to $\vv = \rIm\{
\ol{\phi} \, \vnabla\phi \}$, can be eliminated up to an additive
constant in the transformed Hamiltonian. The minimization of the
energy functional consequently enforces the reality of the
wavefunction $\phi$. More precisely,
\begin{align} \label{eq-0.11}
\cE_{\alpha, \Lambda}\big( \phi, \psi \big)  
\ \geq \ 
\cE_{\alpha, \Lambda}\big(|\phi|, W_\phi \psi \big) \, .
\end{align}
Defining
\begin{align} \label{eq-0.12}
\hcE_{\alpha, \Lambda}\big( \phi \big)  
\ := \ 
\inf\Big\{ \cE_{\alpha, \Lambda}(\phi, \psi)  \; \Big|
\ \psi \in \cF_\ph \, , \ \|\psi\| = 1 \; \Big\} \, ,
\end{align}
we therefore have that
\begin{align} \label{eq-0.13}
\hcE_{\alpha, \Lambda}\big( \phi \big)  
\ \geq \ 
\hcE_{\alpha, \Lambda}\big( |\phi| \big) \, . 
\end{align}

\item[(3)] Eq.~\eqref{eq-0.13} guarantees that we can assume without
  loss of generality that $\phi = |\phi| \geq 0$, and in this case
\begin{align} \label{eq-0.14}
\cE_{\alpha, \Lambda}\big( \phi, \psi \big)  
\ = \ 
\frac{1}{2} \big\| \vnabla \phi \big\|_2^2
+ \bigg\la \psi \bigg| \
\Big( \hf + 
\frac{\alpha}{2} \int |\phi(x)|^2 \, \vA_\Lambda^2(x) \, d^3x \Big) 
\, \psi \bigg\ra_{\cF} \, .
\end{align}
In Theorem~\ref{thm-IV-vb.05} in Section~\ref{subsec-IV.2} we give an
alternative proof for the observation of Lieb and Loss that
\begin{align} \label{eq-0.15}
\inf\bigg\{ \Big\la \psi & \Big| \
\Big( \hf + 
\frac{\alpha}{2} \int |\phi(x)|^2 \, \vA_\Lambda^2(x) \, d^3x \Big) 
\, \psi \Big\ra_{\cF} 
\; \bigg| \ \psi \in \cF_\ph \, , \ \|\psi\| = 1 \; \bigg\} 
\nonumber \\[1ex] 
\ = \ & 
\frac{1}{2} \Tr\big\{ \sqrt{ -\Delta_x + 2\Theta_{|\phi|,\alpha} \, }
\: - \: \sqrt{ -\Delta_x} \big\} \, , 
\end{align}
where $\Theta_{\phi,\alpha} := \alpha (2\pi)^{-3} 
P_C \chi_{\sigma,\Lambda} (\hphi*)^* (\hphi*) \chi_{\sigma,\Lambda} P_C$,
$\chi_{\sigma,\Lambda} := \bfone[ \sigma \leq -\Delta_x \leq \Lambda^2]$, 
and $P_C := \bfone\left[(\vec{\nabla}_x \cdot) =0 \right]$ is the
projection onto divergence-free vector fields, i.e., vector fields 
in Coulomb gauge. Inserting this into \eqref{eq-0.12}-\eqref{eq-0.13},
we arrive at
\begin{align} \label{eq-0.16,1}
\hcE_{\alpha, \Lambda}\big( \phi \big)  
\ = \ 
\frac{1}{2} \big\| \vnabla \phi \big\|_2^2
+ \frac{1}{2} X(2 \Theta_{\phi,\alpha} ) \, ,
\end{align}
for $\phi = |\phi| \geq 0$, where
\begin{align} \label{eq-0.17}
X(A) \ := \ & \Tr\Big( \sqrt{|k|^2 + A \,} - |k| \Big) 
\quad \text{and} 
\\[1ex] \label{eq-0.17,1}
\Theta_{\phi,\alpha} \ := \ & \frac{\alpha}{(2\pi)^3} \, 
P_C \, \chi_{\sigma,\Lambda} \, \phi(x)^2 \, \chi_{\sigma,\Lambda} \, P_C \, ,
\end{align}
with $\phi(x) \equiv \phi(i\nabla_p)$ denoting the corresponding
Fourier multiplier (with respect to the momentum representation).

\item[(4)] In Section~\ref{sec-V} we introduce the infima
\begin{align} \label{eq-0.23,1}
E_\LL^{(L)}(\alpha, \Lambda) \ := \ &
\inf\Big\{ \widehat{\cE}_{\alpha,\Lambda}(\phi_L) \; \Big| 
\ \phi_L \in Y_L \Big\} \, ,
\\[1ex] \label{eq-0.23,2}
F^{(L)}[\beta] \ := \ &
\inf\Big\{ \cF_\beta(\phi_L) \; \Big| \ \phi_L \in Y_L \Big\} \, ,
\end{align}
of the Lieb-Loss functional $E_\LL(\alpha,\Lambda)(\phi_L)$ and the
auxiliary functional \break $\cF_\beta(\phi_L)$ under variation only over
compactly supported functions $\phi_L \in Y_L := H^1(B(0,L))$ and
compare these infima to $E_\LL(\alpha, \Lambda)$ and $F[\beta]$ by
means of the IMS localization formula. More specifically, we prove in
Theorem~\ref{thm-V-vb.01} that
\begin{align} \label{eq-0.24}
E_\LL^{(L)}(\alpha, \Lambda) \; - \; C \, L^{-2} 
\ \leq \ &
E_\LL(\alpha,\Lambda) \ \leq \ E_\LL^{(L)}(\alpha, \Lambda) \, , 
\\[1ex] \label{eq-0.25}
F^{(L)}[\beta] \; - \; C \, L^{-2} 
\ \leq \ &
F[\beta] \ \leq \ F^{(L)}[\beta] \, , 
\end{align}
for some universal constant $C < \infty$ and all $L >0$.
Consequently, the leading orders of $E_\LL(\alpha,\Lambda)$ and
$F[\beta]$, respectively, are determined by their behavior on
compactly supported functions.

\item[(5)] The fourth step carried out in Sections~\ref{sec-VI} 
and \ref{sec-VII} is to find upper and lower bounds for all
compactly supported $\phi = |\phi| \in Y_L := H^1\big(B(0,L)\big)$ 
on $X(\Theta_{\phi_L,\alpha})$. In Theorem~\ref{thm-VI-vb.01} we prove 
the existence of a universal constant $C < \infty$ such that, for all 
$0 < \eps \leq 1$, $L \geq 1/\Lambda$, and $\phi \in Y_L$,
\begin{align} \label{eq-0.18}
\frac{1}{2} X(2\Theta_{\phi,\alpha}) &
\; - \; 
\sqrt{\frac{4\alpha}{9\pi} \,} \, \Lambda^3 \, \|\phi_L\|_1
\\[1ex] \nonumber
\ \leq \ &
C \big( \eps \, \alpha^{\frac{1}{2}} \, \Lambda^3 + 
\alpha^{\frac{1}{2}} \, \sigma^{\frac{3}{2}} \, \Lambda^{\frac{3}{2}} \big)
\, \|\phi_L\|_1
\; + \;
C \, \eps^{-2} \, \Lambda^2 \, L^{\frac{3}{2}} \, \|\nabla\phi_L\|_2 \, .
\end{align}
This is complemented by the lower bound in 
Theorem~\ref{thm-VII-vb.02} which asserts that, there
exists a universal constant $C <\infty$ such that, 
for all $L \geq 1/\Lambda$ and $\phi \in Y_L$,
\begin{align} \label{eq-0.19}
\frac{1}{2} X(2 \Theta_{\phi,\alpha}) 
\; - \; 
\sqrt{\frac{4\alpha}{9\pi} \,} \, \Lambda^3 \, \|\phi_L\|_1
\ \geq \ 
- C \, \alpha^{\frac{1}{4}} \, \Lambda^{\frac{7}{2}} \, 
L^{\frac{3}{2}} \, \|\phi_L\|_1^{\frac{1}{2}} \, .
\end{align}

\item[(6)] Estimates~\eqref{eq-0.18} and \eqref{eq-0.19} suggest
to compare the functional 
$\hcE_{\alpha,\Lambda}(\phi) = \frac{1}{2} \| \vnabla \phi \|_2^2
+ \frac{1}{2} X(2 \Theta_{\phi,\alpha})$
to 
$\cF_{\beta(\alpha,\Lambda)}(\phi) = \frac{1}{2} \| \vnabla \phi \|_2^2 
+  \beta(\alpha,\Lambda) \|\phi\|_1$ with \break
$\beta(\alpha,\Lambda) := \sqrt{\frac{4\alpha}{9\pi}} \Lambda^3$
which is done in Section~\ref{sec-VIII}.
Indeed, this leads us to introduce the family of auxiliary functionals
$(\cF_\beta)_{\beta >0}$, defined on 
$Y := H^1(\RR^3)\cap L^1(\RR^3) \subset H^1(\RR^3)$ as
\begin{align} \label{eq-0.20}
\cF_\beta(\phi)  
\ := \ &
\frac{1}{2} \big\| \vnabla \phi \big\|_2^2 
+  \, \beta \, \|\phi\|_1 \, ,
\end{align}
and their infima
\begin{align} \label{eq-0.21}
F[\beta] \ := \ &
\inf\big\{ \cF_\beta(\phi) \: \big| \ 
\phi \in Y \, , \ \ \|\phi\|_2 = 1 \big\} \, .
\end{align}
This family of functionals is analyzed by direct methods of the
calculus of variations in detail by the second author in a separate
paper \cite{Hach2020}, and here we describe its properties only
briefly. 
  \begin{compactitem}
  \item For fixed $\beta >0$, the functional $\cF_\beta$ possesses a
    minimizer, which is unique up to translations, nonnegative,
    spherically symmetric and decreasing. In particular, its infimum
    $F[\beta]$ is attained and hence a strictly positive minimum.

  \item For all $\beta >0$ both energy and minimizer are uniquely
    determined by their scaling behaviour in $\beta$ and universal
    constants corresponding to the case $\beta = 1$. In particular,
    $F[1] >0$ is a universal positive number and $F[\beta] =
    \beta^{4/7} F_1$.

  \item The Euler-Lagrange equation, which corresponds to the
    inhomogeneous Helmholtz equation $(- \Delta - \mu^2) \phi + \beta
    = 0$, yields an explicit characterization of this minimizer in
    terms of the zeroth Bessel function $j_0$ of the first kind.
   \end{compactitem}
In Section~\ref{sec-VIII} we use the information on the auxiliary
functional and especially the scaling relation 
$F[\beta] = \beta^{4/7} F_1$ to finally derive \eqref{eq-0.08},
formulated again as \eqref{eq-VIII-vb.03} in 
Theorem~\ref{thm-VIII-vb.1}. In order to simultaneously control
the errors on the right side of \eqref{eq-0.18}
and the localization error of order $\cO(L^{-2})$ we choose
$\eps := \alpha^{4/105} \Lambda^{-4/105}$ and 
$L := \alpha^{17/105} \Lambda^{-88/105}$ and arrive at the upper bound
in \eqref{eq-0.08}. Similarly, we choose 
$L := \alpha^{9/49} \Lambda^{-40/49}$ to obtain the lower bound in
\eqref{eq-0.08} from \eqref{eq-0.19} and the localization estimate.

Note that both estimates suggest that the length scale $\ell(\Lambda)$
of the particle in the ground state of the Lieb-Loss model is of order
$\ell(\Lambda) \approx \alpha^{\tau-1} \Lambda^\tau$, with 
$\tau = \frac{6}{7} \approx 0.86$.
\end{itemize}

\noindent \textbf{Acknowledgements:} VB gratefully acknowledges useful
discussions with M.~K{\"o}\-nen\-berg, J.~M{\o}ller, and A.~Pizzo.

\newpage
\section{The Lieb-Loss Model} \label{sec-II}
%
\setcounter{equation}{0}
The Lieb-Loss model is a variational model for the study of the ground
state energy of a system containing a single nonrelativistic spinless
particle which is minimally coupled to the quantized radiation
field. The dynamics of such a quantum system is generated by the
Pauli-Fierz Hamiltonian
\begin{align} \label{eq-II-vb.01}
H_{\alpha, \sigma, \Lambda}
\ := \ 
\frac{1}{2} \Big( i\vnabla + \sqrt{\alpha} \vA_{\sigma, \Lambda}(\vx) \Big)^2 
+ H_\ph \, ,
\end{align}
which we define here as a quadratic form on $\rmH^1(\RR^3) \otimes
\cD(N_\ph^{1/2})$, where $\rmH^1(\RR^3) \subseteq L^2(\RR^3)$ is the
Sobolev space of square-integrable functions whose gradient is
square-integrable, as well, and $\cD(N_\ph^{1/2}) \subseteq \fF_\ph$
denotes the subspace of finite photon number expectation value of the
photon Fock space $\fF_\ph$. The latter is the boson Fock space over
the one-photon Hilbert space $\fh$, i.e., it is the orthogonal sum
$\fF_\ph = \bigoplus_{n=0}^\infty \fF_\ph^{(n)}$ of $n$-photon
sectors, where $\fF_\ph^{(0)} := \CC \cdot \Om$ is the one-dimensional
vacuum sector spanned by the normalized vacuum vector $\Om$, and for
$n \geq 1$, the $n$-photon sector $\fF_\ph^{(n)} := \cS_n[\fh_\pol^{\otimes
  n}] \subseteq \fh_\pol^{\otimes n}$, is the subspace of the $n$-fold
tensor product of $\fh_\pol$ of totally symmetric vectors.

The one-photon Hilbert space $\fh_\pol :=L^2(S_{\sigma, \Lambda} \times
\ZZ_2)$ is the space of square-integrable, divergence-free vector
fields $\vk \mapsto \veps(\vk, +) f(k,+) + \veps(\vk, -) f(k,-)$
supported in the momentum shell $S_{\sigma,\Lambda} := \{ \vk \in
\RR^3: \; \sigma \leq |\vk| < \Lambda \} \subseteq \RR^3$ which
excludes momenta of magnitude below the infrared cutoff $\sigma \geq
0$ and above the ultraviolet cutoff $1 \leq \Lambda < \infty$. The two
transversal polarizations are parametrized by the polarization vectors
$\veps(k,\pm) \perp k$ that are chosen so as to form an orthonormal
frame $\big( \vk/|\vk|, \veps(\vk,+), \veps(\vk,-) \big)$ in $\CC
\otimes \RR^3$, for all $\vk \in S_{\sigma,\Lambda} \setminus
\{\vec{0}\}$. Of course, the map $k \to \veps(k)$ is assumed to be
measurable and, for convenience, chosen to be real, $\veps(k,\pm) \in \RR^3$,
almost everywhere in $\RR^3 \times \ZZ_2$. 

In \eqref{eq-II-vb.01} the field Hamiltonian
\begin{align} \label{eq-II-vb.02} 
H_\ph \ = \ 
\mathrm{d}\Gamma(|k|) 
\ = \ 
\int |k| \,
a^*(k)  \, a(k)  \; dk
\end{align}
represents the energy of the radiation field, and 
\begin{align} \label{eq-II-vb.03} 
\vA_{\sigma, \Lambda}(\vx) \ = \ 
(2\pi )^{-\frac{3}{2}} \int \frac{\veps(k)}{|k|^{\frac{1}{2}}} 
\big( a^*(k) \, e^{-i \vk \cdot \vx} \: + \: a(k) \, e^{i \vk \cdot \vx} \big) 
\: dk 
\end{align}
is the quantized vector potential (in Coulomb gauge). In
\eqref{eq-II-vb.02}, \eqref{eq-II-vb.03}, we denote
elements of $S_{\sigma,\Lambda} \times \ZZ_2 \ni (\vk, \tau)$ by $k :=
(\vk, \tau)$ and then further $-k := (-\vk, \tau)$, $|k| := |\vk|$,
$\int F(k) \, dk := \sum_{\tau = \pm} \int_{\sigma \leq |\vk| < \Lambda}
F(\vk,\tau) \: d^3k$. Furthermore,
we use creation and annihilation operators $a^*(k)$ and $a(k)$, for $k
\in S_{\sigma,\Lambda} \times \ZZ_2$, in \eqref{eq-II-vb.02} and
\eqref{eq-II-vb.03}. These are operator-valued distributions
constituting a Fock representation of the canonical commutation
relations (CCR) on $\fF_\ph$, i.e.,
\begin{align} \label{eq-II-vb.04} 
\big[ a(k_1), \, a(k_2) \big] 
\ = \ 
\big[ a^*(k_1), & \, a^*(k_2) \big] 
\ = \ 0 \; , 
\\[1ex] \label{eq-II-vb.05} 
\big[ a(k_1), \, a^*(k_2) \big] \ = \ \delta(k_1-k_2) 
\; , & \quad 
a(k_1) \Om \ = \ 0 \, ,
\end{align}
for all $k_1, k_2 \in S_{\sigma,\Lambda} \times \ZZ_2$ (integrated
over $k_1$ and $k_2$ against test functions). Finally, the photon
number operator entering the definition of the domain
$\cD(N_\ph^{1/2})$ is given by $N_\ph := \int a^*(k) a(k) \, dk$.

The Lieb-Loss model is defined by the Lieb-Loss (energy) functional
$\cE_{\alpha, \sigma, \Lambda}: \rmH^1(\RR^3) \times \cD(N_\ph^{1/2})
\to \RR$ which results from varying only over products 
$\phi \otimes \psi$ of normalized wave functions of the particle 
$\phi \in L^2(\RR^3)$ and the photon state $\psi \in \fF_\ph$ in the
Rayleigh-Ritz principle, i.e.,
\begin{align} \label{eq-II-vb.06} 
\cE_{\alpha, \sigma, \Lambda}(\phi, \psi)  
\ := \ 
\big\la \phi \otimes \psi \: \big| \ 
H_{\alpha, \sigma, \Lambda} (\phi \otimes \psi) \big\ra \, .
\end{align}
Note that, given a fixed $\phi \in \rmH^1(\RR^3)$ and varying only
over $\psi \in \cD(N_\ph^{1/2})$, the Lieb-Loss functional $\psi
\mapsto \cE_\LL(\phi,\psi)$ becomes the expectation value in $\psi$ of
a Hamiltonian that is quadratic in the boson fields. More
specifically, a simple computation shows that
\begin{align} \label{eq-II-vb.07} 
\cE_{\alpha, \sigma, \Lambda}(\phi, \psi)  
\ = \ &
\frac{1}{2} \big\| \vnabla \phi \big\|_2^2
+ \Big\la \psi \Big| \;
\HH\big( |\phi|^2 , \; \rIm\{ \ol{\phi} \, \vnabla\phi \} \big) 
\: \psi \Big\ra_{\fF} \, ,
\end{align}
where $\la \cdot | \cdot \ra_\fF$ denotes the scalar product on the
photon Fock space $\fF_\ph$ and, for fixed
$\rho: \RR^3 \to \RR^+$ and $\vv: \RR^3 \to \RR^3$, the
quadratic Hamiltonian $\HH[\rho, \vv]$ is given as
\begin{align} \label{eq-II-vb.08}
\HH[\rho, \vv]
\ := \ 
H_\ph + \frac{\alpha}{2} \int \rho(x) \, \vA_{\sigma,\Lambda}^{\: 2}(x) \, d^3x 
+ \sqrt{\alpha} \int \vv(x) \cdot \vA_{\sigma,\Lambda}(x) \, d^3x \, .
\end{align}
As we show below it turns out that the minimal values of the Lieb-Loss
functional is attained for positive wave functions. To exhibit this we
define $r := |\phi| \in \rmH^1(\RR^3; \RR_0^+)$ and choose $\gamma \in
\rmH^1(\RR^3; \RR)$, for a given $\phi \in \rmH^1(\RR^3; \CC)$, so
that 
\begin{align} \label{eq-II-vb.09} 
\phi \ = \ r \, e^{i\gamma} \, ,  
\qquad 
|\phi|^2 \ = \ & r^2 \, ,  
\qquad  
\rIm\{ \ol{\phi} \, \vnabla\phi \} \ = \ r^2 \, \vnabla\gamma \, ,
\\[1ex] \label{eq-II-vb.10} 
\|\vnabla\phi\|_2^2 
\ = \ & 
\|\vnabla r\|_2^2 + \|r \vnabla\gamma\|_2^2 \, ,
\end{align}
and thus
\begin{align} \label{eq-II-vb.11} 
\cE_{\alpha, \sigma, \Lambda}(r \, e^{i\gamma} \; , \; \psi)  
\ = \ &
\frac{1}{2} \|\vnabla r\|_2^2 + \frac{1}{2} \|r \vnabla\gamma\|_2^2
+ \Big\la \psi \Big| \;
\HH\big( r^2 , \; r^2 \, \vnabla\gamma \big) 
\: \psi \Big\ra_{\fF} \, .
\end{align}
Although convenient, the explicit parametrization of Couloumb gauge by
polarization vectors $\veps(\vec{k}, \pm)$ tends to obscure the
picture by introducing a seeming dependence of the model on the choice
of $\veps(\vec{k}, \pm)$, which, however, should be physically
meaningless. For this reason we choose the one-photon space to be the
Hilbert space
\begin{align} \label{eq-II-vb.12}
\fh \ := \ &
P_C \big[ L^2(S_{\sigma, \Lambda}; \CC \otimes \RR^3) \big] 
\\[1ex] \nonumber
\ = \ &
\Big\{ f \in L^2(S_{\sigma, \Lambda}; \CC \otimes \RR^3) \; \Big| \ 
\forall \, \vk \in S_{\sigma, \Lambda}: \ \ \vk \perp f(\vk) \Big\} 
\end{align}
of divergence-free, square-integrable vector fields, where 
$P_C \in \cB\big[ L^2(S_{\sigma, \Lambda}; \CC \otimes \RR^3) \big]$
is the orthogonal projection acting as 
$\big[P_C f \big](\vk) := P_{\vk}^\perp f(\vk)$, with 
$P_{\vk}: \RR^3 \to \RR^3$ being the projection in $\RR^3$ onto the unit 
vector $\vk/\|\vk| \in \mathbb{S}^2$. Note that for
any arbitrary, but fixed, choice of polarization vectors
basis $\{ \veps(\vk, +), \veps(\vk, -) \}_{\vk \in S_{\sigma, \Lambda}}$ 
described above, the map
\begin{align} \label{eq-II-vb.13}
\Xi: \ \fh_\pol \ \to \ \fh \, , \quad
\big[ \Xi f \big](\vk) \ := \ 
\veps(\vk,+) \, f(\vk,+) + \veps(\vk,-) \, f(\vk,-)
\end{align}
is unitary, with $[\Xi^{-1} f](\vk, \pm) = [\Xi^* f](\vk, \pm) =
\veps(\vk,\pm) \cdot f(\vk)$, and allows us to switch between the
photon representations, if necessary.

Accordingly, the photon Fock space we use is 
$\fF_\ph := \fF_b[\fh]$ the bosonic Fock space over divergence-free
vector fields. On $\fF_\ph$ we have a Fock representation of the CCR
of the form
\begin{align} \label{eq-II-vb.14} 
\big[ a(\vk_1, \nu_1), \, a(\vk_2, \nu_2) \big] 
\ = \ 
\big[ a^*(\vk_1, \nu_1), & \, a^*(\vk_2, \nu_2) \big] 
\ = \ 0 \; , 
\\[1ex] \label{eq-II-vb.15} 
\big[ a(\vk_1, \nu_1), \, a^*(\vk_2, \nu_2) \big] 
\ = \ 
\delta(\vk_1-\vk_2) \, \big( P_{\vk_1}^\perp \big)_{\nu_1,\nu_2}
\; , & \quad 
a(k_1) \Om \ = \ 0 \, ,
\end{align}
for all $\vk_1, \vk_2 \in S_{\sigma,\Lambda}$ and $\nu_1, \nu_2 \in \ZZ_3$,
as operator-valued distributions, or
\begin{align} \label{eq-II-vb.16} 
\big[ a(f), \, a(g) \big] 
\ = \ 
\big[ a^*(f), & \, a^*(g) \big] 
\ = \ 0 \; , 
\\[1ex] \label{eq-II-vb.17} 
\big[ a(f), \, a^*(g) \big] 
\ = \ 
\big\la f \, \big| \; P_C \, g \big\ra \; , & \quad 
a(f) \Om \ = \ 0 \, ,
\end{align}
for all $f, g \in \fh$, where we write
\begin{align} \label{eq-II-vb.18} 
a^*(f) \ := \ \sum_{\nu=1}^3 \int f_\nu(\vk) \, a^*(\vk, \nu) \: d^3k
\; , \quad 
a(f) \ := \ \sum_{\nu=1}^3 \int \ol{f_\nu(\vk)} \, a(\vk, \nu) \: d^3k \, .
\end{align}
for all $f = (f_1, f_2, f_3)^t \in \fh$. In this representation
the operator $\vA_{\sigma,\Lambda}(x)$ of the magnetic vector 
potential becomes 
$\vbbA(x) = \big( \bbA_{1}(x) \, , \, \bbA_{2}(x) \, , \, \bbA_{3}(x) \big)$,
with
\begin{align} \label{eq-II-vb.20} 
\bbA_{\mu}(x) 
\ = \ & 
a^*\big( m_\mu(x) \big) + a\big( m_\mu(x) \big) 
\\ \nonumber 
\ = \ & 
\sum_{\nu=1}^3 \int \Big\{
m_{\mu,\nu}(x, \vk) \, a^*( \vk, \nu) + 
\ol{m_{\mu,\nu}(x, \vk)} \, a( \vk, \nu) \Big\} d^3k \, ,
\\[1ex] \label{eq-II-vb.21}
m_{\mu,\nu}(x, \vk) 
& \ := \ 
\frac{\bfone\big[ \sigma \leq |\vk| < \Lambda \big]}{
(2\pi)^{3/2} \, |\vk|^{1/2}} \, \big( P_\vk^\perp \big)_{\mu,\nu} 
\, e^{-i \vk \cdot \vx} \, ,
\end{align}
and the Hamiltonian $\HH\big( r^2 , \; r^2 \, \vnabla\gamma \big)$ 
in \eqref{eq-II-vb.11} turns into
\begin{align} \label{eq-II-vb.22} 
\tHH & \big( r^2 , \; r^2 \, \vnabla \gamma \big) 
\ = \ 
\\[1ex] \nonumber
& H_\ph + \frac{\alpha}{2} \int \big( r(x) \, \vbbA(x) \big)^2 \, d^3x 
+ \sqrt{\alpha} \int \big( r(x) \, \vnabla\gamma(x) \big) \cdot 
\big( r(x) \, \vbbA(x) \big) \, d^3x \, .
\end{align}
Note that the dependence of $\vbbA(x)$ on the cutoff parameters
$0 < \sigma \leq 1$ and $1 \leq \Lambda < \infty$ is not displayed 
anymore.

\newpage
\section{Bogolubov Transformations} \label{sec-III}
%
\setcounter{equation}{0}
Next, we analyze the infimum of $\psi \mapsto \big\la \psi \, \big| \;
\HH\big[ r^2 , r^2 \vnabla\gamma \big] \psi \big\ra$, as 
$\psi \in \cD(N_\ph^{1/2})$ varies over normalized states, by means of
Bogolubov transformations. For a suitable definition of these
in the present context, the choice of the antilinear involution 
$J :\fh \to \fh$ defined by
\begin{align} \label{eq-III-vb.01}
[ J f ](\vk) \ := \ \ol{f(-\vk)} 
\end{align}
plays a key role. Before using $J$, we recall a few facts about
antiunitary maps and generalized creation and annihilation operators.

\subsection{Antiunitary Maps and Generalized Field Operators} 
\label{subsec-III.1}
%
For a general complex Hilbert space $\sfh$ the
Riesz map $\sfR: \sfh \to \sfh^*$, 
$\psi \mapsto \la \psi |$ is a canonical isomorphism
from $\sfh$ onto its dual $\sfh^* = \cB[\sfh; \CC]$.
Moreover, $\sfR$ is \textit{antiunitary}, i.e., it obeys
$\la \sfR(f) | \sfR(g) \ra_{\sfh^*} 
= \la g | f \ra_{\sfh}$. Note that $\sfR$ is not the only
antiunitary map from $\sfh$ to $\sfh^*$, for if 
$\sfu: \sfh \to \sfh$
and $\sfv: \sfh^* \to \sfh^*$ are unitary operators
on $\sfh$ and $\sfh^*$, respectively, then 
$\sfR \circ \sfu: \sfh \to \sfh^*$ and 
$\sfv \circ \sfR: \sfh \to \sfh^*$ are 
antiunitary, too. Conversely, any antiunitary from $\sfh$ to
$\sfh^*$ is of this form.

In the present paper we prefer to work with an \textit{antiunitary}
$\sfJ$ which additionally constitutes an \textit{antilinear
  involution} or \textit{real structure}. Given a general complex
Hilbert space $\sfh$ these are antiunitary bijections $\sfJ: \sfh \to
\sfh$, which obey
\begin{align} \label{eq-III-vb.02}
\sfJ^2 \ = \ \bfone_{\sfh} 
\quad \textit{and} \quad
\forall \, f,g \in \sfh: \ \ 
\la \sfJ(f) | \sfJ(g) \ra_{\sfh} 
\ = \ 
\la g | f \ra_{\sfh} \, .
\end{align}
Given an \textit{antiunitary involution} $\sfJ: \sfh \to \sfh$ we can
define the maximal $\sfJ$-invariant subspace
\begin{align} \label{eq-III-vb.03}
\sfh_\RR 
\ = \ 
\big\{ f \in \sfh \; \big| \ \sfJ f = f \; \big\} 
\ \subseteq \ \sfh \, ,
\end{align}
which is a $\RR$-linear subspace of $\sfh$. Writing
$f \in \sfh$ as $f = f_1 + i f_2$, with
$f_1 := \frac{1}{2}(f + \sfJ f) \in \sfh$ and 
$f_2 := \frac{1}{2i}(f - \sfJ f) \in \sfh$, we obtain
a direct sum decomposition 
$\sfh = \sfh_\RR \oplus i\sfh_\RR$.
Similar to antiunitary operators $\sfh \to \sfh^*$, antiunitary
involutions $\sfh \to \sfh$ are not unique. This gives us freedom to
make a suitable choice for the problem to solve, namely,
\eqref{eq-III-vb.01} in the present case.

To define Bogolubov transformations it is convenient to use
\textit{generalized creation and annihilation operators} which were
first introduced by Araki and Shiraishi in \cite{Araki1971,
  ArakiShiraishi1971} to describe the second quantization of one-body
Hamiltonians. Bogolubov transformations are also discussed in detail
in \cite{Solovej2014, BachLiebSolovej1994}. Given an antiunitary
involution $\sfJ: \sfh \to \sfh$, the generalized creation and
annihilation (field) operators $\cresf, \annsf: \sfh \oplus \sfh \to
\cB[\cD(\sfN^{1/2}); \fF_b(\sfh)]$ are defined by
\begin{align} \label{eq-III-vb.04} 
\cresf(f \oplus \sfJ g) \ := \ a^*(f) + a(g) 
\quad \text{and} \quad
\annsf(f \oplus \sfJ g) \ := \ a(f) + a^*(g) \, ,
\end{align}
for any $f, g \in \sfh$. Note that
\begin{align} \label{eq-III-vb.05} 
\annsf(F) \ = \ \cre(\cJ F) \, ,
\quad \text{with} \quad
\cJ \ := \ \begin{pmatrix} 0 & \sfJ \\ \sfJ & 0 \end{pmatrix} 
\end{align}
being an antiunitary involution on $\sfh \oplus \sfh$. The vectors
in $\sfh \oplus \sfh$ which are invariant under $\cJ$ are of the
form $y \oplus \sfJ y$, with $y \in \sfh$. They form a real subspace
\begin{align} \label{eq-III-vb.06} 
(\sfh \oplus \sfh)_\cJ 
\ := \
\big\{ G \in \sfh \oplus \sfh \: \big| \ G = \cJ G \: \big\}
\ = \
\big\{ y \oplus \sfJ y \: \big| \ y \in \sfh \: \big\}
\ = \
q[\sfh] \, ,
\end{align}
where $\sfq : \sfh \to (\sfh \oplus \sfh)_\cJ$ is the real-linear map
\begin{align} \label{eq-III-vb.07} 
\sfq \ := \ \begin{pmatrix} \bfone \\ \sfJ \end{pmatrix} \, ,
\quad \text{with adjoint} \quad
\sfq^*: (\sfh \oplus \sfh)_\cJ \to \sfh \, , \ 
\sfq^* \ = \ \big( \bfone \, , \, \sfJ \big) \, .
\end{align}
One advantage of the generalized formalism consists in encoding all
orderings in the second quantization of operators, so that we need not
worry about imposing normal-ordering. The price for this is the
slightly modified form of the canonical commutation relations (CCR),
the generalized field operators obey, namely,
\begin{align} \label{eq-III-vb.08} 
\big[ \annsf(F) \, , \, \cresf(F') \big]
\ = \ 
\big\la F \: \big| \ \cS F' \big\ra \, ,
\end{align}
where $\cS$ is a natural symplectic form on $\sfh \oplus \sfh$ given by
\begin{align} \label{eq-III-vb.09} 
\cS \ := \ \begin{pmatrix} 1 & 0 \\ 0 & -1 \end{pmatrix} \, .
\end{align}
%

\subsection{Second Quantization and 
Bogolubov Transformations} \label{subsec-III.2}
%
Next, we introduce the \ul{\textit{second quantization}} of one-photon
operators. Let $\sfJ: \sfh \to \sfh$ be an antiunitary involution and
$\{F_i\}_{i = 1}^\infty \subseteq \sfh \oplus \sfh$ an orthonormal
basis.  For $T = T^* \in \cB[\sfh \oplus \sfh]$ and $y \in \sfh$, 
we define their second quantization 
$\dGamma_\sfJ[T, y] \in \cB[\cD(N_\ph); \fF_\ph)]$ by
\begin{align} \label{eq-III-vb.10} 
\dGamma_\sfJ[T, y]
\ := \ &
\sum_{i,j=1}^\infty \la F_i | \: T F_j \ra \: \cresf(F_i) \, \annsf(F_j) 
\\[1ex] \nonumber 
& + \sum_{i=1}^\infty \big\{ \la F_i | q(y) \ra \, \cresf(F_i) 
+ \ol{\la F_i | \: q(y) \ra} \, \annsf(F_i) \big\} \, . 
\end{align}
Note that the definition \eqref{eq-III-vb.10} of $\dGamma_\sfJ[T, y]$
is independent of the choice of the orthonormal basis 
$\{F_i\}_{i = 1}^\infty \subseteq \sfh \oplus \sfh$. Moreover,
$\dGamma_\sfJ[T, y]$ is self-adjoint on $\cD(N_\ph)$ and
$\dGamma_\sfJ[T, y]$ is semibounded, provided $T \geq 0$. Finally, 
$[a(f), a(g)] = 0$ and $[a^*(f), a^*(g)] = 0$ imply that 
$\dGamma_\sfJ\big[ 
\big(\begin{smallmatrix} a & b \\ c & d \end{smallmatrix} \big), y \big]
= \dGamma_\sfJ\big[ 
\big(\begin{smallmatrix} a & \sfJ b^* \sfJ \\ 
\sfJ c^* \sfJ & d \end{smallmatrix} \big), y \big]$, 
and we can and will henceforth always assume that 
\begin{align} \label{eq-III-vb.10,1} 
b^* \ = \ \sfJ \, b \, \sfJ \, , 
\quad \text{for} \quad
T \ = \ T^* \ = \ 
\big( \begin{smallmatrix} a & b \\ b^* & d \end{smallmatrix} \big) 
\ = \ 
\big( \begin{smallmatrix} a & b \\ \sfJ b \sfJ & d 
\end{smallmatrix} \big) \, .  
\end{align}

A second advantage of the generalized creation and annihilation
operators is that their use eases the definition of Bogolubov
transformations. We recall that \ul{\textit{Bogolubov
    transformations}} are unitary transformations $\widehat{\bbU}$ on
Fock space $\fF_\ph$ which preserve \eqref{eq-III-vb.07} and are
linear in the field operators, i.e., they act as
\begin{align} \label{eq-III-vb.11} 
\widehat{\bbU} \, a^*(f) \, \widehat{\bbU}^* 
\ := \ 
a^*(Uf) + a(\sfJ Vf) + \la \eta | f \ra \, ,
\end{align}
for all $f \in \sfh$, where $U$ and $V$ are linear operators on $\sfh$
and $\eta \in \sfh$. The Bogolubov transformations form a group which
is the semidirect product of the group of \textit{homogenous Bogolubov
  transformations} and the group of \textit{Weyl transformations}. 
That is, every Bogolubov transformation $\widehat{\bbU}$ can be written 
as a composition 
\begin{align} \label{eq-III-vb.12} 
\widehat{\bbU} \ = \  \bbU_\qB \, \bbW_\eta \ = \ \bbW_\mu \, \bbU_\qB 
\end{align}
of a homogeneous Bogolubov transformation $\bbU_\qB$ and a Weyl transformation
$\bbW_\eta$ or a composition of a Weyl transformation
$\bbW_\mu$ and $\bbU_\qB$, but with $\mu \neq \eta$, in general.

\ul{\textit{Homogeneous Bogolubov transformations}} $\bbU_\qB$ are the
special case $\eta = 0$ of \eqref{eq-III-vb.11}. In terms of the
generalized field operators they assume the form
\begin{align} \label{eq-III-vb.13} 
\bbU_\qB \, \cre(F) \, \bbU_\qB^* \ := \ \cre(\qB F)
\, , \quad  
\qB \ \equiv \  \qB(U,V) 
\ := \ 
\begin{pmatrix}
U & \sfJ V \sfJ \\ V & \sfJ U \sfJ  
\end{pmatrix} \, ,
\end{align}
where the form of $\qB$ is determined by \eqref{eq-III-vb.05}, i.e.,
$\cJ \qB = \qB \cJ$, and \eqref{eq-III-vb.11}. Note that this makes
explicit use of the antiunitary involution $\sfJ: \sfh \to \sfh$. The
homogeneous Bogolubov transformation $\bbU_\qB$ is unitary iff it
leaves the CCR invariant and preserves the norm of the vacuum vector
$\Om \in \fF_\ph$, which is equivalent to
\begin{align} \label{eq-III-vb.14} 
\qB^* \, \cS \, \qB \ = \ \cS \ , \quad
\qB \, \cS \, \qB^* \ = \ \cS \ , \quad
\quad \text{and} \quad 
\Tr\big( V^* \, V \big) \ < \ \infty \, .
\end{align}
The second identity in \eqref{eq-III-vb.14} is actually a consequence
of the first, as the latter implies the invertibility of $B$, and then
the second identity follows from the uniqueness of the inverse. The
requirement that $V$ be a Hilbert-Schmidt operator is known as the
\textit{Shale-Stinespring condition}. A simple computation shows that
the second quantization $\dGamma_\sfJ[T, y]$ of $T$ and $y$
transforms under a homogeneous Bogolubov transformation $\bbU_\qB$ with
$\qB \equiv \qB(U,V)$ as 
\begin{align} \label{eq-III-vb.15} 
\bbU_\qB \, \dGamma_\sfJ[ T, \; y ] \, \bbU_\qB^*  
\ = \ 
\dGamma_\sfJ\big[ \qB T \qB^* \, , \; \tfrac{1}{2} q^* \qB q(y) \big] \, .
\end{align}

\ul{\textit{Weyl transformations}} $\bbW_\eta$ are the special case $U
= \bfone_\sfh$ and $V = 0$ of \eqref{eq-III-vb.11}. They act on the
generalized field operators as
\begin{align} \label{eq-III-vb.16} 
\bbW_\eta \, \cre(F) \, \bbW_\eta^* 
\ := \ 
\cre(F) + \la q(\eta) \, | \: F \ra \, .
\end{align}
The unitarity of $\bbW_\eta$ is equivalent to the requirement $\eta
\in \sfh$. Another simple computation shows that the second
quantization $\dGamma_\sfJ[T, y]$ of $T$ and $y$ transforms under a
Weyl transformation $\bbW_\eta$ as
\begin{align} \label{eq-III-vb.17} 
\bbW_\eta \, \dGamma_\sfJ[ T, \; y] \, \bbW_\eta^*  
\ = \ 
\dGamma_\sfJ\Big[ T \, , \ y + \tfrac{1}{2}q^* & T q(\eta) \Big] 
+ \la \eta | \; q^* T q(\eta) \ra + 4 \rRe \la \eta | y \ra \, .
\end{align}
%

\subsection{The Lieb-Loss Model in Terms of Second 
Quantization} 
\label{subsec-III.4}
%
We turn to the analysis of the Lieb-Loss model. Note that
$\dGamma_\sfJ[T, y]$ depends on the choice of the antiunitary
involution $\sfJ: \fh \to \fh$. For the analysis of the Lieb-Loss model
it is of key importance to choose the antiunitary involution $J: \fh
\to \fh$ as
\begin{align} \label{eq-III-vb.18} 
\forall \, f \in \fh, \ \vk \in S_{\sigma,\Lambda}: 
\quad [Jf](\vk) \ := \ \ol{f(-\vk)} 
\end{align}
because with this choice the operator 
$T: \fh_\RR \oplus \fh_\RR \to \fh_\RR \oplus \fh_\RR$ leaves the
real subspace $\fh_\RR \oplus \fh_\RR$ of $\fh \oplus \fh$ invariant,
and the vector $y \in \fh_\RR$ is contained in the real subspace
$\fh_\RR \subseteq \fh$ of $J$-invariant vectors, as is discussed
below.

We identify $\tHH( r^2 , \, r^2 \vnabla \gamma )$ with
$\dGamma_J[T_{r,\alpha}, y_{r,\gamma,\alpha}]$, for suitably chosen
$T_{r,\alpha}$ and $y_{r,\gamma,\alpha}$. We state the result in form
of Lemma~\ref{lem-III-vb.02} below.
%
\begin{lemma} \label{lem-III-vb.02}
Let $J: \fh \to \fh$ be defined by \eqref{eq-III-vb.18} and 
$r, \gamma \in H^1(\RR^3)$. Then the Lieb-Loss functional 
\eqref{eq-II-vb.11} is given by
\begin{align} \label{eq-III-vb.19} 
\cE_{\alpha, \sigma, \Lambda}(r \, e^{i\gamma} \; , \; \psi)  
\ = \ &
\frac{1}{2} \|\vnabla r\|_2^2 + \frac{1}{2} \|r \vnabla\gamma\|_2^2
+ \frac{1}{2} \Big\la \psi \Big| \;
\dGamma_J\big[ T_{r,\alpha} \, , \; y_{r,\gamma,\alpha} \big] 
\: \psi \Big\ra_{\fF} \, ,
\end{align}
where 
\begin{align} \label{eq-III-vb.20}
T_{r,\alpha} \ := \ 
|k|^{-1/2} \, \begin{pmatrix}
2|k|^2 + \Theta_{r,\alpha} & \Theta_{r,\alpha} \\ 
\Theta_{r,\alpha} & \Theta_{r,\alpha} 
\end{pmatrix} \, |k|^{-1/2} \, , 
\end{align}
with $|k|$ denoting the multiplication operator 
$\big[ |k| f \big](\vk) := |k| f(\vk)$ (Fourier multiplier), 
and $\Theta_{r,\alpha}$ being a nonnegative, $J$-invariant, self-adjoint 
Hilbert-Schmidt operator,
$\Theta_{r,\alpha} = \Theta_{r,\alpha}^* = 
\Theta_{r,\alpha}^T = J \Theta_{r,\alpha} J \geq 0$
given by
\begin{align} \label{eq-III-vb.21,1} 
\Theta_{r,\alpha} \ = \ & \Phi_{r,\alpha}^* \, \Phi_{r,\alpha} 
\, , \quad 
\Phi_{r,\alpha} \ = \ 
(\hr *) \, P_C \, \chi_{\sigma,\Lambda} \, ,
\\[1ex] \label{eq-III-vb.21,2} 
\Phi_{r,\alpha}(\vp,\mu \, ; & \, \vk, \nu)
\ := \ \alpha^{1/2} \, (2\pi)^{-3/2} \,
\hr(\vp-\vk) \, \big(P_\vk^\perp \big)_{\mu,\nu} 
\, \chi_{\sigma,\Lambda}(\vk) \, ,
\end{align}
where $\big[ \chi_{\sigma,\Lambda} f \big](\vk) \: := \: 
\bfone[ \sigma \leq |\vk| < \Lambda ] \, f(\vk)$ is a
multiplication operator, and 
$\hr *$ is the convolution operator  
$[\hr * f](\vk) = \int \hr(\vk - \vk') \, f(\vk') \: d^3k'$,
where $\hr \equiv \cF[r]$ denotes the Fourier transform
$\cF[r](\vk) \: := \: (2\pi)^{-3/2} \int e^{-i\vk \cdot x} \, r(x) \, d^3x$
of $r$, normalized as to preserve the $L^2$-scalar product.
\\
Furthermore, $y_{r,\gamma,\alpha} = J[y_{r,\gamma,\alpha}] \in \fh_\RR$ is given by
\begin{align} \label{eq-III-vb.22} 
y_{r,\gamma,\alpha} \ = \ & 
|k|^{-1/2} \, \Phi_{r,\alpha}^* \, \cF[ r \, \vnabla\gamma ]
\quad \Leftrightarrow 
\\ \nonumber 
y_{r,\gamma,\alpha}(\vk,\nu)
\ := \ &
\sum_{\mu=1}^3 \int |\vk|^{-1/2} \, 
\Phi_{r,\alpha}^*\big(\vk,\nu; \vp,\mu\big) \, 
\cF\big[ r \, \partial_\mu\gamma \big](\vp) \: d^3p \, .
\end{align}
\end{lemma}
\begin{proof}
We first observe that
\begin{align} \label{eq-III-vb.23} 
\frac{\alpha}{2} \int \big( r(x) & \, \vbbA(x) \big)^2 \, d^3x 
\ = \
\sum_{\mu = 1}^3 \int \frac{\alpha}{2} \Big[ 
a^*\big( r(x) \, m_\mu(x) \big) + a\big( r(x) \, m_\mu(x) \big) 
\Big]^2 \, d^3x
\nonumber \\[1ex] 
\ = \ &
\sum_{\mu = 1}^3 \frac{\alpha}{2} \int
\cre\Big( q\big[r(x) \, m_\mu(x)\big] \Big) 
\: \ann\Big( q\big[r(x) \, m_\mu(x)\big] \Big) \; d^3x
\nonumber \\[1ex] 
\ = \ &
\frac{1}{2} \, \dGamma_J\left[ |k|^{-1/2} \,  \begin{pmatrix} 
\Theta_{r,\alpha} & \Theta_{r,\alpha} \\ \Theta_{r,\alpha} & \Theta_{r,\alpha} 
\end{pmatrix} \, |k|^{-1/2} \; , \ 0  \right] \, ,
\end{align}
where $\Theta_{r,\alpha}: \fh \to \fh$ is the bounded operator given by the
integral kernel
\begin{align} \label{eq-III-vb.24} 
|\vk|^{-1/2} \, \Theta_{r,\alpha}(\vk,\nu \, ; & \, \vk', \nu') \, |\vk|^{-1/2}
\ := \ 
\sum_{\mu = 1}^3 \int \alpha \, r^2(x) \: 
m_{\mu,\nu}(x,\vk) \, m_{\mu,\nu'}(x,\vk') \: d^3x \, ,
\end{align}
recalling the definition $m_{\mu,\nu}(x, \vk) 
\, := \, (2\pi)^{-3/2} |\vk|^{-1/2} \chi_{\sigma,\Lambda}(\vk) 
\big( P_\vk^\perp \big)_{\mu,\nu} e^{-i \vk \cdot x}$ from \eqref{eq-II-vb.21}. 
As $J(e^{i \vk \cdot \vx} e_\nu) = e^{i \vk \cdot \vx} e_\nu$ we have that 
$J[r(x) m_\mu(x)] = r(x) m_\mu(x)$ and hence
\begin{align} \label{eq-III-vb.25} 
\Theta_{r,\alpha} \ = \ J \, \Theta_{r,\alpha} 
\ = \ 
\Theta_{r,\alpha} \, J \ = \ J \, \Theta_{r,\alpha} \, J \, .
\end{align}
Moreover, using the Plancherel theorem, we have that 
\begin{align} \label{eq-III-vb.25,1} 
\Theta_{r,\alpha} \ = \ \Phi_{r,\alpha}^* \, \Phi_{r,\alpha} \, ,
\end{align}
where $\Phi_{r,\alpha} = \alpha^{1/2} (2\pi)^{-3/2} 
(\hr *) P_C \chi_{\sigma,\Lambda} |\vk|^{-1/2}$ 
is defined by the integral kernel
\begin{align} \label{eq-III-vb.25,2} 
\Phi_{r,\alpha}(\vp,\mu \, ; & \, \vk, \nu)
\ := \ 
\frac{\alpha^{1/2}}{(2\pi)^{3/2}} \, \hr(\vp-\vk) \, 
\big(P_\vk^\perp \big)_{\mu,\nu} \, \chi_{\sigma,\Lambda}(\vk) \, ,
\end{align}
i.e., $\hr *$ is the convolution operator  
$[\hr * f](\vk) = \int \hr(\vk - \vk') \, f(\vk') \: d^3k'$,
convolving $f$ with the Fourier transform
\begin{align} \label{eq-III-vb.28} 
\cF[r](\vk) \ \equiv \ \hr(\vk) 
\ := \ 
\int e^{-i\vk \cdot x} \, r(x) \, \frac{d^3x}{(2\pi)^{-3/2}}
\end{align}
of $r$, normalized as to preserve the $L^2$-scalar product.
\\
Similarly, we obtain
\begin{align} \label{eq-III-vb.29} 
& \alpha^{1/2} \int r^2(x) \, \vnabla\gamma(x) \cdot \vbbA(x) \, d^3x 
\nonumber \\[1ex] 
&  = 
\sum_{\mu = 1}^3 \int \alpha^{1/2}  
\Big\{ a^*\big( r^2(x) \, \partial_\mu\gamma(x) \, m_\mu(x) \big) 
+ a\big( r^2(x) \, \partial_\mu\gamma(x) \, m_\mu(x) \big) \Big\} \, d^3x
\nonumber \\[1ex] 
&  =  
\sum_{\mu = 1}^3 \int \frac{\alpha^{1/2}}{2}  
\Big\{ \cre\big( q[r^2(x) \, \partial_\mu\gamma(x) \, m_\mu(x)] \big) 
+ \ann\big( q[r^2(x) \, \vnabla\gamma(x) \cdot \vm(x)] \big) \Big\} \, d^3x
\nonumber \\[1ex] 
&  = 
\frac{1}{2} \, \dGamma_J\big[ 0 \; , \ y_{r,\gamma,\alpha} \big] \, ,
\end{align}
where $y_{r,\gamma,\alpha} \in \fh$ is given as
\begin{align} \label{eq-III-vb.30} 
y_{r,\gamma,\alpha}(\vk,\nu)
\ := \ &
\sum_{\mu = 1}^3 \int r^2(x) \: \partial_\mu \gamma(x) \: \alpha^{1/2} \, 
m_{\mu,\nu}(x,\vk) \: d^3 x \, .
\end{align}
Note that $J m_\mu(x) = m_\mu(x)$ implies 
$y_{r,\gamma,\alpha} = J[y_{r,\gamma,\alpha}] \in \fh_\RR$ and the 
Plancherel theorem yields
$y_{r,\gamma,\alpha} = |k|^{-1/2} \, \Phi_{r,\alpha}^* \cF[ r \, \vnabla\gamma ]$, 
i.e., 
\begin{align} \label{eq-III-vb.30,1} 
y_{r,\gamma,\alpha}(\vk,\nu)
\ := \ 
\sum_{\mu=1}^3 \int |\vk|^{-1/2} \, \Phi_{r,\alpha}^*\big(\vk,\nu; \vp,\mu\big) \, 
\cF\big[ r \, \partial_\mu\gamma \big](\vp) \: d^3p \, .
\end{align}
\end{proof}

\newpage
\section{Minimization over Photon States} \label{sec-IV} 
%
\setcounter{equation}{0}
%
\subsection{Weyl Transformations and \\ Positivity of the 
Electron Wave Function} \label{subsec-IV.1}
%
In this section we show that the optimal electron wave function is
nonnegative. More precisely, given any normalized complex-valued
electron wave function $\phi \in H^1(\RR^3)$, we show that the
Lieb-Loss functional for the electron wave function $|\phi| \in
H^1(\RR^3)$ yields a lower value, if minimized over all photon
states. This is done by a suitable Weyl transformation that eliminates
the term in the Hamiltonian which is linear in the field
operators. The proper choice \eqref{eq-III-vb.18} of the antiunitary
$J$ is of key importance for the construction of this Weyl
transformation. Equally important is the observation, that
the energy shift induced by this Weyl transformation is balanced
by the term $\frac{1}{2}\|r \vnabla\gamma\|_2^2$ that vanishes
for real $\phi$.

We start with a preparatory lemma.
%
\begin{lemma} \label{lem-IV-vb.01}
Let $\kappa \in \cB[\fh]$ be a bounded operator and $\delta \in \RR^+$.
Then
\begin{align} \label{eq-IV-vb.00,1} 
\kappa \, \big( \delta^2 + \kappa^* \kappa \big)^{-1} \, \kappa^* 
\ \leq \ 
\bfone \, .
\end{align}
\end{lemma}
\begin{proof}
Note that, if $\kappa$ is invertible, the assertion follows trivially
from the operator monotonicity of $A \mapsto A^{-1}$, namely,
\begin{align} \label{eq-IV-vb.00,2} 
\kappa \, \big( \delta^2 + \kappa^* \kappa \big)^{-1} \, \kappa^* 
\ \leq \ 
\kappa \, \big( \kappa^* \kappa \big)^{-1} \, \kappa^* 
\ = \ 
\kappa \, \kappa^{-1} \, \big( \kappa^* \big)^{-1} \, \kappa^* 
\ = \
\bfone \, .
\end{align}
If $\kappa$ is, however, not invertible then we define the bounded
operator $A \in \cB[\fh \oplus \fh]$ by 
\begin{align} \label{eq-IV-vb.00,3} 
A \ := \ 
\begin{pmatrix} 
\delta    & \kappa^* \\ 
\kappa & -\delta 
\end{pmatrix} 
\ = \ A^* \, ,
\end{align}
and observe that
\begin{align} \label{eq-IV-vb.00,4} 
A^2 \ = \ 
\begin{pmatrix} 
\delta^2 + \kappa^* \kappa & 0 \\ 
0  & \delta^2 + \kappa \kappa^* 
\end{pmatrix} 
\ =: \ M \ \geq \ \delta^2 \cdot \bfone  
\end{align}
clearly is invertible. Hence
\begin{align} \label{eq-IV-vb.00,5} 
A^2 \, M^{-1} \ = \ M^{-1} \, A^2 \ = \ \bfone \, ,
\end{align}
and $A$ has a left inverse $M^{-1} A$ and a right inverse $A M^{-1}$.
Thus $A$ is invertible and its left and right inverses coincide.
In particular, 
\begin{align} \label{eq-IV-vb.00,6} 
& \begin{pmatrix} 
\bfone & 0 \\ 
0  & \bfone 
\end{pmatrix} 
\ = \
A \, M^{-1} \, A \ = \ 
\\[1ex] \nonumber
& \begin{pmatrix} 
(\delta^2 + \kappa^* \kappa)^{-1} \delta^2 + 
\kappa^* (\delta^2 + \kappa \kappa^*)^{-1} \kappa  
& \delta (\delta^2 + \kappa^* \kappa)^{-1} \kappa^* 
- \kappa^* (\delta^2 + \kappa \kappa^*)^{-1} \delta \\ 
\kappa (\delta^2 + \kappa^* \kappa)^{-1} \delta
- \delta (\delta^2 + \kappa \kappa^*)^{-1} \kappa 
& (\delta^2 + \kappa \kappa^*)^{-1} \delta^2 + 
\kappa (\delta^2 + \kappa^* \kappa)^{-1} \kappa^*
\end{pmatrix} \, .  
\end{align}
Evaluating the lower right corner, we obtain
\begin{align} \label{eq-IV-vb.00,7} 
\bfone \ = \ 
(\delta^2 + \kappa \kappa^*)^{-1} \delta^2 + 
\kappa (\delta^2 + \kappa^* \kappa)^{-1} \kappa^*
\ \geq \ 
\kappa^* (\delta^2 + \kappa^* \kappa)^{-1} \kappa \, . 
\end{align}
\end{proof}

\begin{lemma} \label{lem-IV-vb.02}
Let $J: \fh \to \fh$ be defined by \eqref{eq-III-vb.18},
$r, \gamma \in H^1(\RR^3)$, and $T_{r,\alpha}$, $\Theta_{r,\alpha}$, and 
$y_{r,\gamma,\alpha} \in \fh_\RR$ as in 
\eqref{eq-III-vb.20}-\eqref{eq-III-vb.22}. Then there is a unique
$\eta_{r,\gamma} \in \fh_\RR$ such that 
\begin{align} \label{eq-IV-vb.01}
y_{r,\gamma,\alpha} 
\ = \ 
\frac{1}{2} \, q^* \, T_{r,\alpha} \, q (\eta_{r,\gamma}) \, . 
\end{align}
Moreover, as a quadratic form
\begin{align} \label{eq-IV-vb.02} 
\dGamma_J\big[ T_{r,\alpha} , \; y_{r,\gamma,\alpha} \big] 
\ \geq \
\bbW_{\eta_{r,\gamma}} \, 
\dGamma_J\big[ T_{r,\alpha} , \; 0 \big] 
\, \bbW_{\eta_{r,\gamma}}^* \: - \: \| r \, \vnabla\gamma \|_2^2 \, .
\end{align}
\end{lemma}
\begin{proof}
We first compute that
\begin{align} \label{eq-IV-vb.05}
q^* \, T_{r,\alpha} \, q 
\ = \ &
|k|^{-1/2} \, ( \bfone \ J ) \,  
\begin{pmatrix}
2|k|^2 + \Theta_{r,\alpha} & \Theta_{r,\alpha} \\ \Theta_{r,\alpha} & \Theta_{r,\alpha} 
\end{pmatrix} \,  
\begin{pmatrix}
\bfone \\ J 
\end{pmatrix} \, |k|^{-1/2} 
\nonumber \\[1ex]
\ = \ &
|k|^{-1/2} \, \big( 2|k|^2 + \Theta_{r,\alpha} + J \Theta_{r,\alpha} 
+ \Theta_{r,\alpha} J + J \Theta_{r,\alpha} J \big) \, |k|^{-1/2} \, ,
\end{align}
so $\eta_{r,\gamma}$ sought for fulfils
\begin{align} \label{eq-IV-vb.06}
2 y_{r,\gamma,\alpha} 
\ = \ 
|k|^{-1/2} \, \big( 2|k|^2 + \Theta_{r,\alpha} + J \Theta_{r,\alpha} 
+ \Theta_{r,\alpha} J + J \Theta_{r,\alpha} J \big) 
\, |k|^{-1/2} \, \eta_{r,\gamma} \, .
\end{align}
If $J$ was any general antiunitary map, the determination of
$\eta_{r,\gamma}$ from \eqref{eq-IV-vb.06} appeared to be fairly
complicated, but thanks to our choice \eqref{eq-III-vb.18} of $J$ we
have that $\Theta_{r,\alpha} = J \Theta_{r,\alpha} = \Theta_{r,\alpha} J 
= J \Theta_{r,\alpha} J$ and
$y_{r,\gamma,\alpha} = J y_{r,\gamma,\alpha}$. Therefore,
$y_{r,\gamma,\alpha}$ is an element of $\fh_\RR$ which is left invariant 
by $q^* T_{r,\alpha} q = |k|^{-1/2} \big( 2|k|^2 + 4 \Theta_{r,\alpha}
\big) |k|^{-1/2}$.  Moreover, 
$q^* T_{r,\alpha} q \geq 2 |k| \geq 2 \sigma \cdot \bfone >0$ is
strictly positive and hence invertible, due to 
$\Theta_{r,\alpha} \geq 0$. (Here, the infrared cutoff $\sigma >0$
comes in handy.) It follows that
\begin{align} \label{eq-IV-vb.07}
\eta_{r,\gamma} 
\ = \ 
|k|^{1/2} \, \big( |k|^2 + 2 \Theta_{r,\alpha} \big)^{-1} 
\, |k|^{1/2} \, y_{r,\gamma,\alpha}
\ \in \ \fh_\RR 
\end{align}
and 
\begin{align} \label{eq-IV-vb.08}
\big\la \eta_{r,\gamma} \; \big| \ 
|k|^{-1/2} & \big( |k|^2 + 2 \Theta_{r,\alpha} \big) \,
|k|^{-1/2} \, \eta_{r,\gamma} \big\ra
\nonumber \\[1ex] 
\ = \ &
\big\la |k|^{1/2} \, y_{r,\gamma,\alpha} \; \big| \ 
\big( |k|^2 + 2 \Theta_{r,\alpha} \big)^{-1} \, 
|k|^{1/2} \, y_{r,\gamma,\alpha} \big\ra 
\nonumber \\[1ex]
\ = \ &
\Big\la \cF[r \, \vnabla\gamma] \; \Big| \ \Phi_{r,\alpha} \,
\big( |k|^2 + 2 \Phi_{r,\alpha}^* \Phi_{r,\alpha} \big)^{-1} \, 
\Phi_{r,\alpha}^* \, \cF[r \, \vnabla\gamma] \Big\ra 
\nonumber \\[1ex]
\ \leq \ &
\Big\la \cF[r \, \vnabla\gamma] \; \Big| \; \cF[r \, \vnabla\gamma] \Big\ra 
\ = \ 
\| r \, \vnabla\gamma \|_2^2 \, ,
\end{align}
estimating $\big( |k|^2 + 2 \Phi_{r,\alpha}^* \Phi_{r,\alpha} \big)^{-1} \leq
\big( \sigma^2 + 2 \Phi_{r,\alpha}^* \Phi_{r,\alpha} \big)^{-1}$ and then
using Lemma~\ref{lem-IV-vb.01}. We obtain the assertion from here
by \eqref{eq-III-vb.17}.
\end{proof}

As a corollary of Lemma~\ref{lem-IV-vb.02}, we now find the
following lower bound on the Lieb-Loss functional defined in
\eqref{eq-0.04}.
%
\begin{corollary} \label{cor-IV-vb.03}
Let $\phi \in H^1(\RR^3)$ and $\psi \in \fF_\ph$ be normalized
wave functions. Then there exists a unitary Weyl transformation
$\bbW_\phi$ such that
\begin{align} \label{eq-IV-vb.09}
\cE_{\alpha, \Lambda}\big( \phi, \psi \big)  
\ \geq \ 
\cE_{\alpha, \Lambda}\big(|\phi|, \bbW_\phi \psi \big) \, .
\end{align}
\end{corollary}
%
As a consequence it follows that the partial minimization 
of the Lieb-Loss functional
\begin{align} \label{eq-IV-vb.10}
\hcE_{\alpha, \Lambda}\big( \phi \big)  
\ := \ 
\inf\Big\{ \cE_{\alpha, \Lambda}(\phi, \psi)  \; \Big|
\ \psi \in \cF_\ph \, , \ \|\psi\| = 1 \; \Big\} \, ,
\end{align}
over photon wave functions [see~\eqref{eq-0.12}] allows us
to restrict the minimization over electron wave functions
to nonnegative functions.
%
\begin{theorem} \label{thm-IV-vb.04}
Let $J: \fh \to \fh$ be defined by \eqref{eq-III-vb.18} 
and suppose that $\phi \in \cH_\el$ is normalized and 
$\phi \in H^1(\RR^3)$. Then
\begin{align} \label{eq-IV-vb.11}
\hcE_{\alpha, \Lambda}\big( \phi \big) 
\ \geq \ 
\hcE_{\alpha, \Lambda}\big( |\phi| \big) 
\ = \ 
\frac{1}{2} \big\|\vnabla |\phi| \big\|_2^2 
+ \frac{1}{2} \inf\big\{
\sigma\big( \dGamma_J[ T_{|\phi|,\alpha} \, , \, 0 ] \big) \big\} \, ,
\end{align}
where $\sigma(A) \subseteq \RR$ denotes the spectrum of a self-adjoint
operator $A$ and $T_{|\phi|}$ is as defined in 
\eqref{eq-III-vb.20}-\eqref{eq-III-vb.21,2}.
\end{theorem}

%
\subsection{The Ground State Energy of $T_{|\phi|,\alpha}$} 
\label{subsec-IV.2}
%
In this section we show that the infimum of the spectrum
of $\frac{1}{2} \dGamma_J[ T_{|\phi|} , \, 0 ]$ equals 
$X(\Theta_{|\phi|,\alpha})$, as defined in \eqref{eq-0.17} and
\eqref{eq-III-vb.20}-\eqref{eq-III-vb.21,2}. This fact had already
been observed in \cite{LiebLoss1999}, and we give an alternative and
detailed proof here. More specifically, we prove the following theorem
in this section.
%
\begin{theorem} \label{thm-IV-vb.05}
Let $J: \fh \to \fh$ be defined by \eqref{eq-III-vb.18},
suppose that $\phi = |\phi| \in H^1(\RR^3)$, and 
let $T_{\phi,\alpha}$ and $\Theta_{\phi,\alpha}$ be given as in 
\eqref{eq-III-vb.20}-\eqref{eq-III-vb.22}. Then 
\begin{align} \label{eq-IV-vb.52}
\inf\big\{
\sigma\big( \dGamma_J[ T_{\phi,\alpha} , 0 ] \big) \big\} 
\ = \ 
\Tr\Big( \sqrt{ |k|^2 + 2 \Theta_{\phi,\alpha} \:} \: - \: |k| \Big) \, .
\end{align}
\end{theorem}
%
Inserting \eqref{eq-IV-vb.52} into \eqref{eq-IV-vb.11}, 
we immediately obtain the following Corollary.
%
\begin{corollary} \label{cor-IV-vb.06}
Let $J: \fh \to \fh$ be defined by \eqref{eq-III-vb.18} 
and suppose that $\phi = |\phi| \in \cH_\el$ is normalized and 
$\phi \in H^1(\RR^3)$. Then
\begin{align} \label{eq-IV-vb.52,1}
\hcE_{\alpha, \Lambda}(\phi)
\ = \ 
\frac{1}{2} \big\|\vnabla \phi \big\|_2^2 
+ \frac{1}{2} \Tr\Big( \sqrt{ |k|^2 + 2\Theta_{\phi,\alpha} \:} 
\: - \: |k| \Big) \, .
\end{align}
where $\Theta_{\phi,\alpha}$ is defined in 
\eqref{eq-III-vb.21,1}-\eqref{eq-III-vb.21,2}.
\end{corollary}
%
\begin{proof} \textit{[Proof of Theorem~\ref{thm-IV-vb.05}]} 
The first step in our proof rests on an observation 
made in \cite{BachBreteauxTzaneteas2013} that, given a nonnegative
Hamiltonian $\HH$ representing an interacting quantum system, it holds true
that
\begin{align} \label{eq-IV-vb.52,2}
\inf_{\rho \in \QF} \big\{ \Tr(\rho^{1/2} \HH \rho^{1/2}) \big\} 
\ = \
\inf_{\rho \in \QF} \big\{ \Tr(\rho^{1/2} \HH \rho^{1/2}) \; \big| 
\ \text{$\rho$ is pure} \big\} \, ,
\end{align}
where $\QF$ denotes the set of quasifree
density matrices. In other words, for the computation of the
Bogolubov-Hartree-Fock energy of the system, one may restrict the
variation over all quasifree states to pure states. 
This statement may be viewed as a generalization of Lieb's variational
principle \cite{Lieb1981a}.  
In Lemma~\ref{lem-IV-vb.07} below, the
observation from \cite{BachBreteauxTzaneteas2013} is applied to the
Hamiltonian $\dGamma_J[ T_{\phi,\alpha} , 0 ]$ and yields the statement, that
its ground state energy is the lowest vacuum expectation value of all
homogeneous Bogolubov transforms of $\dGamma_J[ T_{\phi,\alpha} , 0 ]$,
\begin{align} \label{eq-IV-vb.53}
\inf\big\{
\sigma\big( \dGamma_J[ T_{\phi,\alpha} , 0 ] \big) \big\} 
\ = \ 
\inf\Big\{ \big\la \Om \big| 
\bbU_B \, \dGamma_J[ T_{\phi,\alpha} , 0 ] \, \bbU_B^* \Om \big\ra
\; \Big| \ 
B \in \Bog_J[\fh] \; \Big\} \, ,
\end{align}
where $\Bog_J[\fh]$ is defined in \eqref{eq-IV-vb.12,1}.
\\
Next, an application of Lemma~\ref{lem-IV-vb.08} with $a := 2|k|$, $b
:= |k|^{-1/2} \Theta_{\phi,\alpha} |k|^{-1/2}$, and $d := 0$ yields the
following lower bound on the vacuum expectation values on the right of
\eqref{eq-IV-vb.53} in terms of $|v|$, where $v \in \cL^2[\fh]$ is the
lower left matrix entry of $B$ of the Bogolubov transformation
$\bbU_B$,
\begin{align} \label{eq-IV-vb.54}
& \big\la \Om \big| 
\bbU_B \, \dGamma_J[ T_{\phi,\alpha} , 0 ] \, \bbU_B^* \Om \big\ra 
\ \geq \ 
\\[1ex] \nonumber
& 
\inf_{v \in \cL^2[\fh], \, v \geq 0} \Big\{ 
\Tr\Big[ 2 |k|^{1/2} \, v^2 \, |k|^{1/2} \: + \: 
\Theta_{\phi,\alpha}^{1/2}  |k|^{-1/2} \big(v - \sqrt{1+v^2} \big)^2 
|k|^{-1/2} \Theta_{\phi,\alpha}^{1/2} \Big] \Big\} \, .
\end{align}
The infimum on the right side of the lower bound \eqref{eq-IV-vb.54}
is explicitly computed in Lemma~\ref{lem-IV-vb.09} below, using  
$\sigma \cdot \bfone \leq a := 2|k| \leq \Lambda \cdot \bfone$, 
$b := |k|^{-1/2} \Theta_{\phi,\alpha} |k|^{-1/2} \geq 0$, and $d := 0$ again.
Consequently, 
\begin{align} \label{eq-IV-vb.55}
\inf\Big\{ \big\la \Om \big| 
\bbU_B \, \dGamma_J[ T_{\phi,\alpha} , 0 ] \, \bbU_B^* \Om \big\ra
\; \Big| \ 
B \in \Bog_J[\fh] \; \Big\} 
\ \geq \
\Tr\Big[ \sqrt{ k^2 + 2 \Theta_{\phi,\alpha} \,} \: - \: |k| \Big] \, .
\end{align}
We finally define 
\begin{align} \label{eq-IV-vb.56}
B_* \ := \ 
\begin{pmatrix} 
\sqrt{1+v_*^2 \,} & -v_* \\
-v_* &  \sqrt{1+v_*^2 \,} \\
\end{pmatrix} 
\ = \
\frac{1}{2} \begin{pmatrix} 
y_*^{1/2} + y_*^{-1/2}  & -y_*^{1/2} + y_*^{-1/2} \\
-y_*^{1/2} + y_*^{-1/2} &  y_*^{1/2} + y_*^{-1/2} \\
\end{pmatrix} \, , 
\\[1ex] \label{eq-IV-vb.57}
v_* \ := \ \frac{1}{2} \big( y_*^{1/2} - y_*^{-1/2} \big) \ \geq \ 0 \, , 
\qquad
y_* \ := \ 
|k|^{-1/2} \,  \sqrt{ k^2 + 2 \Theta_{\phi,\alpha} \,} \, |k|^{-1/2} 
\ \geq \ 
\bfone \, ,
\end{align}
in accordance with \eqref{eq-IV-vb.35} and \eqref{eq-IV-vb.45}. Then,
by \eqref{eq-IV-vb.50}, $\Theta_{\phi,\alpha} \in \cL^2[\fh]$ implies that
$y_*-1 \in \cL^2[\fh]$ which is equivalent to $v_* \in \cL^2[\fh]$,
thanks to \eqref{eq-IV-vb.36c}, and thus 
$1-y_*^{-1} = y_* - 1 - 4v_*^2 \in \cL^2[\fh]$. 
Moreover, as $|k|$ and $\Theta_{\phi,\alpha}$
are $J$-invariant, so are $y_*$ and hence also $v_*$ and 
$\sqrt{1+v_*^2 \,}$. It follows that $B_* \in \Bog_J[\fh]$
is a homogeneous Bogolubov transformation. Finally, 
\begin{align} \label{eq-IV-vb.58}
B_* & \, T_{\phi,\alpha} \, B_*^* 
\\[1ex] \nonumber
\ = \ &
B_* \, \begin{pmatrix} 2|k| & 0 \\ 0 & 0 \end{pmatrix} \, B_*^*
\: + \: B_* \, 
\begin{pmatrix} 
|k|^{-1/2} \Theta_{\phi,\alpha} |k|^{-1/2} & |k|^{-1/2} \Theta_{\phi,\alpha} |k|^{-1/2} \\
|k|^{-1/2} \Theta_{\phi,\alpha} |k|^{-1/2} & |k|^{-1/2} \Theta_{\phi,\alpha} |k|^{-1/2} \\
\end{pmatrix} \, B_*^*
\\[1ex] \nonumber
\ = \ &
\frac{1}{2}
\begin{pmatrix} 
(y_*^{1/2} + y_*^{-1/2}) \, |k| \, (y_*^{1/2} + y_*^{-1/2})
& - (y_*^{1/2} + y_*^{-1/2}) \, |k| \, (y_*^{1/2} - y_*^{-1/2}) 
\\
+ 2 y_*^{-1/2} |k|^{-1/2} \Theta_{\phi,\alpha} |k|^{-1/2} y_*^{-1/2} 
& + 2 y_*^{-1/2} |k|^{-1/2} \Theta_{\phi,\alpha} |k|^{-1/2} y_*^{-1/2} 
\\[1ex]
- (y_*^{1/2} - y_*^{-1/2}) \, |k| \, (y_*^{1/2} + y_*^{-1/2})
& (y_*^{1/2} - y_*^{-1/2}) \, |k| \, (y_*^{1/2} - y_*^{-1/2}) 
\\
+ 2 y_*^{-1/2} |k|^{-1/2} \Theta_{\phi,\alpha} |k|^{-1/2} y_*^{-1/2} 
& + 2 y_*^{-1/2} |k|^{-1/2} \Theta_{\phi,\alpha} |k|^{-1/2} y_*^{-1/2} \\
\end{pmatrix} \! ,
\end{align}
so 
\begin{align} \label{eq-IV-vb.59}
\big\la \Om \big| 
& \bbU_{B_*} \, \dGamma_J[ T_{\phi,\alpha} , 0 ] \, \bbU_{B_*}^* \Om \big\ra
\ = \
\big\la \Om \big| \, 
\dGamma_J[ B_* T_{\phi,\alpha} B_*^*, 0 ] \, \Om \big\ra
\nonumber \\[1ex] 
\ = \ &
\frac{1}{2} \, \Tr\Big[
(y_*^{1/2} - y_*^{-1/2}) \, |k| \, (y_*^{1/2} - y_*^{-1/2}) 
+ 2 y_*^{-1/2} |k|^{-1/2} \Theta_{\phi,\alpha} |k|^{-1/2} y_*^{-1/2} \Big]
\nonumber \\[1ex] 
\ = \ &
\frac{1}{2} \, \Tr\Big[ |k|^{1/2} (y_* + y_*^{-1} - 2) |k|^{1/2} +
2 \Theta_{\phi,\alpha}^{1/2} |k|^{-1/2} y_*^{-1} |k|^{-1/2} 
\Theta_{\phi,\alpha}^{1/2} \Big]
\nonumber \\[1ex] 
\ = \ &
\frac{1}{2} \, \Tr\Big[ \sqrt{k^2 + 2 \Theta_{\phi,\alpha} \,} - |k|
+ |k|^{1/2} (y_*^{-1}-1) |k|^{1/2} +
2 \Theta_{\phi,\alpha} |k|^{-1/2} y_*^{-1} |k|^{-1/2}  \Big] 
\nonumber \\[1ex] 
\ = \ &
\frac{1}{2} \, \Tr\Big[ \sqrt{k^2 + 2 \Theta_{\phi,\alpha} \,} - |k|
- \big(k^2 + 2 \Theta_{\phi,\alpha} \big) |k|^{-1/2} y_*^{-1} |k|^{-1/2} 
-|k| \Big]
\nonumber \\[1ex] 
\ = \ &
 \Tr\Big[ \sqrt{k^2 + 2 \Theta_{\phi,\alpha} \,} - |k| \Big] \, . 
\end{align}
\end{proof}
%
The first step in our derivation rests on an observation
made in \cite{BachBreteauxTzaneteas2013} which may be viewed as
a generalization of Lieb's variational principle \cite{Lieb1981a}.
%
\begin{lemma} \label{lem-IV-vb.07}
Let $\sfJ: \sfh \to \sfh$ be an antiunitary involution and
$T = T^* \in \cB[\sfh \oplus \sfh]$ be nonnegative, $T \geq 0$ 
Then 
\begin{align} \label{eq-IV-vb.12}
\inf\big\{
\sigma\big( \dGamma_\sfJ[ T , 0 ] \big) \big\} 
\ = \ 
\inf\Big\{ \big\la \Om \big| 
\bbU_B \, \dGamma_\sfJ[ T , 0 ] \, \bbU_B^* \Om \big\ra
\; \Big| \ 
B \in \Bog_\sfJ[\sfh] \; \Big\} \, ,
\end{align}
where 
\begin{align} \label{eq-IV-vb.12,1}
\Bog_\sfJ[\sfh] \ := \
\bigg\{ B \: = \: 
\begin{pmatrix} 
U & \sfJ V \sfJ \\ V & \sfJ U \sfJ 
\end{pmatrix} 
\ \bigg| \ B^* \cS B = \cS \, , \ \ \Tr(V^*V) < \infty \; \bigg\} 
\end{align}
denotes the set of generators of homogeneous Bogolubov transformations.
\end{lemma}
\begin{proof} Suppose that $\HH \geq 0$ is a nonnegative Hamiltonian 
on $\cF_\ph$ and define its \textit{Bogolubov-Hartree-Fock energy}
by
\begin{align} \label{eq-IV-vb.13}
E_{BHF}(\HH) \ := \ 
\inf\Big\{ \Tr\{ \rho^{1/2} \, \HH \, \rho^{1/2} \} \; \Big|
\ \rho \in \DM \, , \ \text{$\rho$ is quasifree} \Big\} \, ,
\end{align}
where 
$\DM := \big\{ \rho \in \cB[\fF_\ph] \; \big| \; 
0 \leq \rho \leq \Tr\{\rho\} = 1 \big\}$
denotes the set of density matrices on $\fF_\ph$. In
\cite{BachBreteauxTzaneteas2013} it is shown that the
Bogolubov-Hartree-Fock energy is already obtained by taking the
infimum over all \textit{pure} quasifree states,
\begin{align} \label{eq-IV-vb.14}
E_{BHF}(\HH) \ = \ 
\inf\Big\{ \Tr\{ \rho^{1/2} \, \HH \, \rho^{1/2} \} \; \Big|
\ \rho \in \DM \, , \ \text{$\rho$ is quasifree and pure} \Big\} \, .
\end{align}
Since $\dGamma_J[ T_{|\phi|} , \, 0 ]$ is quadratic in the field 
operators, its ground state energy agrees with its 
Bogolubov-Hartree-Fock energy,
\begin{align} \label{eq-IV-vb.15}
\inf\big\{
\sigma\big( \dGamma_\sfJ[ T , 0 ] \big) \big\} 
\ = \ 
E_{BHF}\big( \dGamma_\sfJ[ T , 0 ] \big) \, .
\end{align}
On the other hand, the pure quasifree density matrices
$\rho_{\mathrm{pure}} \in \DM$ are precisely the rank-one
orthogonal projections 
$\rho_{\mathrm{pure}} = 
| \bbU_B^* \bbW_\eta^* \Om \ra \la \bbU_B^* \bbW_\eta^* \Om |$
onto Bogolubov and Weyl transforms $\bbU_B^* \bbW_\eta^* \Om$ of the 
vacuum vector $\Om$, using that, $\bbU_B^* = \bbU_{\cS B^* \cS}$ 
is a homogeneous Bogolubov transformation, for $B \in \Bog_\sfJ[\sfh]$, and
$\bbW_\eta^* = \bbW_{-\eta}$ is a Weyl transformation, for $\eta \in \sfh$.
Thus we obtain
\begin{align} \label{eq-IV-vb.16}
\inf\big\{
\sigma & \big( \dGamma_\sfJ[ T , 0 ] \big) \big\} 
\nonumber \\[1ex] 
\ = \ & 
\inf\Big\{ \big\la \Om \big| 
\bbW_\eta \, \bbU_B \, 
\dGamma_\sfJ[ T , 0 ] \, \bbU_B^* \, \bbW_\eta^* \Om \big\ra
\; \Big| \ B \in \Bog_\sfJ[\sfh] \, , \ \eta \in \sfh \Big\}
\nonumber \\[1ex]
\ = \ &
\inf\Big\{ \big\la \Om \big| 
\dGamma_\sfJ\big[ B T B^* \, , \; 
-\tfrac{1}{2} q^* B T B^* q \eta \big] \Om \big\ra
\\ \nonumber & \qquad \qquad 
+ \la \eta | q^* B T B^* q \eta \ra
\; \Big| \ B \in \Bog_\sfJ[\sfh] \, , \ \eta \in \sfh \Big\} \, ,
\end{align}
using \eqref{eq-III-vb.15} and \eqref{eq-III-vb.17}. Since 
\begin{align} \label{eq-IV-vb.17}
\big\la \Om \big| 
\dGamma_\sfJ\big[ B T B^* , -\tfrac{1}{2} q^* B T B^* q \eta \big] 
\Om \big\ra
\ = \
\big\la \Om \big| 
\dGamma_\sfJ\big[ B T B^* , 0 \big] \Om \big\ra
\end{align}
and
\begin{align} \label{eq-IV-vb.18}
\la \eta | q^* B T B^* q \eta \ra \ \geq \ 0 \, ,
\end{align}
it follows that the infimum on the right side of \eqref{eq-IV-vb.16}
is attained for $\eta =0$.
\end{proof}

\begin{lemma} \label{lem-IV-vb.08}
Let $\sfj: \sfh \to \sfh$ be an antiunitary involution. Let
$a \in \cB[\sfh]$ be a bounded, $b \in \cL^2[\sfh]$ a Hilbert-Schmidt,
and $d \in \cL^1(\sfh)$ a trace-class operator such that
all three are nonnegative and commute with $\sfj$, i.e., 
$a = \sfj a \sfj \geq 0$, $b = \sfj b \sfj \geq 0$, 
$d = \sfj d \sfj \geq 0$. Furthermore let $B \in \Bog_\sfj[\sfh]$,
with $\Bog_\sfj[\sfh]$ as defined in \eqref{eq-IV-vb.12,1}. Then
\begin{align} \label{eq-IV-vb.19}
T \ = \ 
\begin{pmatrix} a + b & b \\ b & d +b \end{pmatrix} 
\ \geq \ 0 \, ,
\end{align}
and 
\begin{align} \label{eq-IV-vb.20}
& \big\la \Om \big| 
\bbU_B \, \dGamma_\sfj[ T , 0 ] \, \bbU_B^* \Om \big\ra
\\[1ex] \nonumber
& \ \geq \ 
\inf\Big\{ 
\Tr\Big[ a \, v^2 + b \big(v - \sqrt{1+v^2} \big)^2 
+ d \, (1+v^2) \Big] \ \Big| 
\ v \geq 0 \, , \ \Tr(v^2) < \infty \Big\} \, .
\end{align}
\end{lemma}
\begin{proof} First, we note that 
\begin{align} \label{eq-IV-vb.20,1}
\text{if} \quad
\tT \ = \ 
\begin{pmatrix} \ta & \tb^* \\ \tb & \td \end{pmatrix} 
\, , \quad \text{then} \quad
\big\la \Om \big| \dGamma_\sfj[ \tT , 0 ] \Om \big\ra
\ = \ \Tr(\td) \, .
\end{align}
Next, if $B \in Bog_\sfj[\sfh]$ is of the form
\begin{align} \label{eq-IV-vb.21}
B \ = \ 
\begin{pmatrix} 
u & \sfj v \sfj \\ v & \sfj u \sfj 
\end{pmatrix} \, ,
\end{align}
then a simple computation using that $\sfj$ commutes with $a$, $b$,
and $d$, shows that
\begin{align} \label{eq-IV-vb.22}
B \, T \, B^* 
\ = \ 
\begin{pmatrix} 
u (a+b) u^* + u b \sfj v^* \sfj & \quad u (a+b) v^* + u b \sfj u^* \sfj \\
+ \sfj v \sfj b u + \sfj v (d+b) v^* \sfj & 
       \quad + \sfj v \sfj b v^* + \sfj v (d+b) u^* \sfj \\[2ex]
v (a+b) u^* + \sfj u \sfj b u^* & \quad v (a+b) v^* + v b \sfj u^* \sfj \\
+ v b \sfj v^* \sfj + \sfj u (d+b) v^* \sfj & 
       \quad + \sfj u \sfj b v^* + \sfj u (d+b) u^* \sfj 
\end{pmatrix} \, .
\end{align}
Using \eqref{eq-IV-vb.20,1}, this yields
\begin{align} \label{eq-IV-vb.23}
\big\la \Om \big| \dGamma_\sfj[ & B^* T B , 0 ] \Om \big\ra
\ = \ 
\Tr\big[ v (a+b) v^* + v^* b \sfj u \sfj 
+ \sfj u^* \sfj b v + u (d+b) u^* \big] 
\\[1ex] \nonumber
\ = \ &
\Tr\big[ a v^* v +  d u^* u \big]
+  2 \rRe\Tr\big[ b v^* \sfj u \sfj \big] \, .
\end{align}
From the Cauchy-Schwarz inequality for traces we obtain
\begin{align} \label{eq-IV-vb.24}
\big| \Tr\big[ b v^* \sfj u \sfj \big] \big|^2 
\ \leq \ &
\Tr\big[ b \, v^* \, x^{-1} \, v \big] \, 
\Tr\big[ b \, \sfj u^* \sfj \, x \, \sfj u \sfj \big] 
\ = \ 
\Tr\big[ b \, v^* \, x^{-1} \, v \big] \, 
\Tr\big[ b \, u^* \sfj \, x \, \sfj u \big] \, ,
\end{align}
for any bounded and invertible positive operator 
$x \geq \mu \cdot \bfone > 0$. 
\\[1ex]
Next we remark that, due to \eqref{eq-III-vb.14}, we have 
\begin{align} \label{eq-IV-vb.25}
u^* \, u - v^* \, v 
\ = \ 
\sfj \, u \, u^* \, \sfj - v \, v^* 
\ = \ \bfone 
\, , \quad
v^* \, \sfj \, u 
\ = 
\ u^* \, \sfj \, v  
\, , \quad 
\sfj \, u \, v^* 
\ = \ 
v \, u^* \, \sfj \, .
\end{align}
For any $r > 0$, this implies that
\begin{align} \label{eq-IV-vb.26}
(r + v v^*) \, \sfj u 
\ = \ &
r \sfj u + v u^* \sfj v
\ = \ 
\sfj u \, (r + v^* v) 
\, , \quad
(r + v v^*) \, v 
\ = \ &
v \, (r + v^* v) \, ,
\end{align}
which, in turn, gives 
\begin{align} \label{eq-IV-vb.27}
(r + v v^*)^{-1} \, \sfj u 
\ = \ 
\sfj u \, (r + v^* v)^{-1} 
\, , \quad \!
(r + v v^*)^{-1} \, v 
\ = \ &
v \, (r + v^* v)^{-1} \, .
\end{align}
Writing the square root as an integral over resolvents according to
$A^{-1/2} = \break \frac{1}{\pi} \int_0^\infty (s + A)^{-1} \, \frac{ds}{s^{1/2}}$,
\eqref{eq-IV-vb.27} yields 
\begin{align} \label{eq-IV-vb.28}
(r + v v^*)^{\pm 1/2} \, \sfj u 
\ = \ 
\sfj u \, (r + v^* v)^{\pm 1/2} 
\, , \quad
(r + v v^*)^{\pm 1/2} \, v 
\ = \ &
v \, (r + v^* v)^{\pm 1/2} \, ,
\end{align}
for all $r >0$. For small $0 < \eps < 1$, we define
\begin{align} \label{eq-IV-vb.29}
x_\eps \ := \ ( 1 + v v^* )^{-1/2} \, ( \eps + v v^* )^{1/2}  
\end{align}
and observe that, due to \eqref{eq-IV-vb.28} and \eqref{eq-IV-vb.25}, 
we have
\begin{align} \label{eq-IV-vb.30}
u^* \sfj \, x_\eps \, \sfj u 
\ = \ &
u^* \sfj \, ( 1 + v v^* )^{-1/2} \, ( \eps + v v^* )^{1/2} \, \sfj u 
\ = \ 
u^* u \, ( 1 + v^* v )^{-1/2} \, ( \eps + v^* v)^{1/2} 
\nonumber \\[1ex]
\ = \ &
( 1 + v^* v )^{1/2} \, ( \eps + v^* v)^{1/2} 
\end{align}
and further
\begin{align} \label{eq-IV-vb.31}
v^* \, x_\eps^{-1} \, v 
\ = \ &
v^* \, ( \eps + v v^* )^{-1/2} \, ( 1 + v v^* )^{1/2} \, v
\ = \ 
v^* v \, ( \eps + v^* v )^{-1/2} \, ( 1 + v^* v)^{1/2} 
\nonumber \\[1ex]
\ \leq \ &
(v^*v)^{1/2} \, ( 1 + v^* v)^{1/2} \, . 
\end{align}
Inserting \eqref{eq-IV-vb.30} and \eqref{eq-IV-vb.31} into
\eqref{eq-IV-vb.24} and taking the limit $\eps \to 0$, we obtain
\begin{align} \label{eq-IV-vb.32}
\big| \Tr\big[ b \, v^* \, \sfj u \sfj \big] \big| 
\ \leq \ 
\Tr\big[ b_\pm \, (v^*v)^{1/2} \, ( 1 + v^* v)^{1/2} \big] \, .
\end{align}
Using this estimate and \eqref{eq-IV-vb.23}, we arrive at 
\begin{align} \label{eq-IV-vb.33}
\big\la \Om \big| \dGamma_\sfj[ B^* T B , 0 ] \Om \big\ra
\ \geq \ 
\Tr\Big[ a \, |v|^2  + b \big( |v| - \sqrt{1 + |v|^2} \big)^2 
+  d (1 + |v|^2) \Big] \, ,
\end{align}
from which the asserted estimate \eqref{eq-IV-vb.20} is immediate.
\end{proof}

\begin{lemma} \label{lem-IV-vb.09}
Let $\sfj: \sfh \to \sfh$ be an antiunitary involution. Let
$a \in \cB[\sfh]$ be a bounded, $b \in \cL^2[\sfh]$ a Hilbert-Schmidt,
and $d \in \cL^1(\sfh)$ a trace-class operator such that
$a = \sfj a \sfj \geq \sigma \cdot \bfone >0$, for some $\sigma >0$,
and $b = \sfj b \sfj \geq 0$, 
$d = \sfj d \sfj \geq 0$, i.e., all three are nonnegative and commute 
with $\sfj$. Then
\begin{align} \label{eq-IV-vb.34}
& \inf\Big\{ 
\Tr\Big[ a \, v^2 + b \big(v - \sqrt{1+v^2} \big)^2 
+ d \, (1+v^2) \Big] \ \Big| 
\ v \geq 0 \, , \ \Tr(v^2) < \infty \Big\} \, .
\nonumber \\[1ex] 
& \ = \ 
\frac{1}{2} \, \Tr\Big[ 
\Big( \sqrt{a+d} \, (a+d + 4b) \, \sqrt{a+d} \Big)^{1/2} - a + d \Big] \, . 
\end{align}
\end{lemma}
\begin{proof} It is convenient to parametrize $v$ as
\begin{align} \label{eq-IV-vb.35}
v \ = \ \frac{1}{2} \big( y^{1/2} - y^{-1/2} \big) \, ,
\end{align}
where $y \geq 1$ is a positive operator defined by \eqref{eq-IV-vb.35}
through functional calculus. Note in passing that $y$ is uniquely
determined by $v$ up to $\ker(y-1)$ and that $y-1 \in \cL^2[\sfh]$,
due to Lemma~\ref{lem-IV-vb.10}~(i). Then
\begin{align} \label{eq-IV-vb.36}
v^2 \ = \ 
\frac{y}{4} + \frac{y^{-1}}{4} - \frac{1}{2}
\quad \text{and} \quad
1+v^2 \ = \ 
\frac{y}{4} + \frac{y^{-1}}{4} + \frac{1}{2} 
\ = \
\Big[\frac{1}{2} \big( y^{1/2} + y^{-1/2} \big) \Big]^2 \, .
\end{align}
Hence we have that
\begin{align} \label{eq-IV-vb.41}
\sqrt{1+v^2} 
\ = \ 
\frac{1}{2} \big( y^{1/2} + y^{-1/2} \big) 
\quad \text{and} \quad
\big( v - \sqrt{1+v^2} \big)^2 \ = \ y^{-1} \, .
\end{align}
Inserting the parametrization \eqref{eq-IV-vb.35} into the
trace in \eqref{eq-IV-vb.34}, we obtain
\begin{align} \label{eq-IV-vb.42}
\Tr\Big[ a \, v^2 + b \big(v - \sqrt{1+v^2} \big)^2 
+ d \, (1+v^2) \Big] 
\ = \ 
\frac{1}{4} \, \cG(y) \, ,
\end{align}
with 
\begin{align} \label{eq-IV-vb.43}
\cG(y) \ := \ \Tr\Big[ 
m^2 \, y \: + \: (m^2 + 4b) \, y^{-1} \: + \:  2 (d-a) \Big] \, ,
\end{align}
$m := \sqrt{a + d} \geq \sigma^{1/2} >0$, and $y-1 \in \cL^2[\fh]$.
Obviously, $y \mapsto \cG(y)$ is convex. We define $y_* \geq \bfone$
by
\begin{align} \label{eq-IV-vb.44}
y_* \ := \ 
m^{-1} \, \big( m \, (m^2 + 4b) \, m \big)^{1/2} \, m^{-1}
\end{align}
and observe that $y_*-1 \in \cL^2[\sfh]$, by Lemma~\ref{lem-IV-vb.10}~(ii),
and that $y_* \, m^2 \, y_* = m^2 + 4b$ which is equivalent to
\begin{align} \label{eq-IV-vb.45}
y_*^{-1} \, (m^2+4b) \, y_*^{-1} \ = \ m^2 \, .
\end{align}
The latter is the formal condition for stationarity of $y \mapsto
\cG(y)$. We refrain from turning this formal into a mathematically
rigorous condition by establishing differentiability of $\cG$ in a
suitable sense. Instead, we simply check by computation that $y_*$ is
the minimizer of $\cG$. Namely, we have 
that 
\begin{align} \label{eq-IV-vb.46}
\cG(y) - \cG(y_*) 
\ := \ 
\Tr\Big[ 
m^2 \, (y-y_*) \: + \: (m^2 + 4b) \, \big(y^{-1} - y_*^{-1} \big) \Big] \, ,
\end{align}
and the second resolvent equation gives
\begin{align} \label{eq-IV-vb.47}
y^{-1} - y_*^{-1}
\ = \ 
- y_*^{-1} \, (y-y_*) \, y_*^{-1} \: + \: 
y_*^{-1} \, (y-y_*) \, y^{-1} \, (y-y_*) \, y_*^{-1} \, .
\end{align}
Thus from \eqref{eq-IV-vb.45} derives
\begin{align} \label{eq-IV-vb.48}
\cG(y) - \cG(y_*) 
\ = \
\Tr\Big[ & 
\big( m^2 - y_*^{-1} \, (m^2+4b) \, y_*^{-1} \big) \, (y-y_*) 
\\ \nonumber  &
\: + \: y^{-1/2} \, (y-y_*) \, y_*^{-1} \, (m^2 + 4b) \, 
y_*^{-1} \, (y-y_*) \, y^{-1/2} \Big] 
\\[1ex] \nonumber 
\ = \
\Tr\Big[ & y^{-1/2} \, (y-y_*) \, y_*^{-1} \, (m^2 + 4b) \, 
y_*^{-1} \, (y-y_*) \, y^{-1/2} \Big] 
\ \geq \ 0 \, .
\end{align}
Finally,
\begin{align} \label{eq-IV-vb.51}
\cG(y_*) \ = \ &
\Tr\Big[ 
m^2 \, y_* \: + \: (m^2 + 4b) \, y_*^{-1} \: + \:  2 (d-a) \Big] 
\nonumber \\[1ex] 
\ = \ &
2 \, \Tr\big[ 
m \, y_* \, m + d-a \big] 
\ = \ 
2 \, \Tr\big[ 
\big( m \, (m^2 + 4b) \, m \big)^{1/2} + d-a \big] 
\nonumber \\[1ex] 
\ = \ &
2 \, \Tr\Big[ 
\big( \sqrt{a+d} \, (a+d + 4b) \, \sqrt{a+d} \big)^{1/2} - a + d \Big] \, , 
\end{align}
arriving at \eqref{eq-IV-vb.34}.
\end{proof}

\begin{lemma} \label{lem-IV-vb.10}
Let $\sfh$ be a Hilbert space and $m, b, y \in \cB[\sfh]$ be positive
bounded operators such that $b \in \cL^2[\sfh]$ is Hilbert-Schmidt,
$y \geq 1$, and $m \geq \sigma^{1/2} \cdot \bfone$, for some 
$\sigma >0$. Then the following assertions hold true.
\begin{itemize}
\item[(i)] Define $v := \frac{1}{2} (y^{1/2} - y^{-1/2}) > 0$. Then 
  $v \in \cL^2[\sfh]$ is Hilbert-Schmidt if, and only if, 
  $y-1 \in \cL^2[\sfh]$ is Hilbert-Schmidt.
\item[(ii)] Define $y := m^{-1} \big[ m (m^2 + 4b) m \big]^{1/2} m^{-1}$.
Then $y \geq 1$ and $y-1 \in \cL^2[\sfh]$ is Hilbert-Schmidt.
\end{itemize}
\end{lemma}
\begin{proof} \hspace*{2mm} \\
\textit{(i):} First $0 < y^{-1/2} \leq 1$ and thus 
$1 \leq y^{1/2} = y^{-1/2} + 2v \leq 1 + 2 \|v\|_\op$, which implies that
\begin{align} \label{eq-IV-vb.36a}
\bfone \ \leq \ y 
\ \leq \ 
\big( 1 + 2\|v\|_{\cB[\sfh]} \big)^2 \cdot \bfone \, .
\end{align}
Secondly note that
\begin{align} \label{eq-IV-vb.36b}
v^2 \ = \ 
\frac{y}{4} + \frac{y^{-1}}{4} - \frac{1}{2}
\ = \ 
\frac{1}{4y} (y-1)^2 \, , 
\end{align}
and taking \eqref{eq-IV-vb.36a} into account, we arrive at \textit{(i)}
because
\begin{align} \label{eq-IV-vb.36c}
\Tr\big[ v^2 \big]
\ \leq \
\Tr\big[ (y-1)^2 \big] 
\ \leq \
4 \big( 1 + 2\|v\|_{\cB[\sfh]} \big)^2 \, \Tr\big[ v^2 \big] \, .
\end{align}
\textit{(ii):} 
For $y = m^{-1} \big[ m (m^2 + 4b) m \big]^{1/2} m^{-1}$ we trivially 
have $y \geq 1$ since $b \geq 0$ and the square root is operator
monotone. Moreover, using 
\begin{align} \label{eq-IV-vb.36d}
A^{-1/2} \ = \ 
\frac{1}{\pi} \int_0^\infty (s + A)^{-1} \, \frac{ds}{s^{1/2}} \, ,
\end{align}
the second resolvent equation, and $R_s := (s + m^4)^{-1} \leq
(s + \sigma^2)^{-1}$, we have that
\begin{align} \label{eq-IV-vb.49}
y - & 1 \ = \ 
m^{-1} \Big[ \big( m^4 + 4mbm \big)^{1/2} \: - \:  m^2 \Big] \, m^{-1} 
\nonumber \\[1ex] 
\ = \ &
\frac{1}{\pi} \int_0^\infty m^{-1} \, \bigg\{ 
\frac{m^4 + 4mbm}{s + m^4 + 4mbm} - \frac{m^4}{s + m^4} 
\bigg\} \, m^{-1} \, \frac{ds}{s^{1/2}}
\nonumber \\[1ex] 
\ = \ &
\frac{1}{\pi} \int_0^\infty m^{-1} \, \Big\{ 
\big( s + m^4 \big)^{-1} - \big( s + m^4 + 4mbm \big)^{-1} 
\Big\} \, m^{-1} \, s^{1/2} \, ds
\nonumber \\[1ex] 
\ = \ &
\frac{4}{\pi} \int_0^\infty \Big\{ R_s \, b \, R_s \: - \:
4 R_s \, bm \, \big( s + m^4 + 4mbm \big)^{-1} \, mb \, R_s \Big\} 
\, s^{1/2} \, ds
\nonumber \\[1ex] 
\ \leq \ &
\frac{4}{\pi} \int_0^\infty \big\{ R_s \, b \, R_s \big\} 
\, s^{1/2} \, ds \, .
\end{align}
Consequently, 
\begin{align} \label{eq-IV-vb.50}
\Tr\big[ & (y-1)^2 \big] 
\nonumber \\[1ex] 
\ \leq \ &
\frac{16}{\pi^2} \int_0^\infty \int_0^\infty \Tr\big[ 
\sqrt{R_s} \, \sqrt{R_t} \, b \, \sqrt{R_t} \, R_s \,  
\sqrt{R_t} \, b \, \sqrt{R_t} \, \sqrt{R_s} \big] 
\, \sqrt{s} \, \sqrt{t} \, ds \, dt
\nonumber \\[1ex] 
\ \leq \ &
\frac{16}{\pi^2} \int_0^\infty \int_0^\infty \Tr\big[ 
\sqrt{R_s} \, \sqrt{R_t} \, b^2 \, \sqrt{R_t} \, \sqrt{R_s} \big] 
\, \frac{\sqrt{s} \,  ds}{s+\sigma^2} \, 
\frac{\sqrt{t} \, dt}{t+\sigma^2} 
\\[1ex] \nonumber 
\ \leq \ &
\frac{16}{\pi^2} 
\bigg( \int_0^\infty \frac{\sqrt{s} \,  ds}{(s+\sigma^2)^2} \bigg)^2 \, 
\Tr\big[ b^2 \big] 
\ = \ 
\frac{16}{\pi^2 \, \sigma^2} 
\bigg( \int_0^\infty \frac{\sqrt{r} \,  ds}{(r+1)^{4}} \bigg)^2 \, 
\Tr\big[ b^2 \big] 
\ < \ \infty \, .
\end{align}
\end{proof}

\newpage
\section{Localization Estimates} \label{sec-V}
%
\setcounter{equation}{0}
In this section we turn to the analysis of the effective energy
functional
\begin{align} \label{eq-V-vb.01}
\hcE_{\alpha, \Lambda}(\phi)
\ = \ 
\frac{1}{2} \big\|\vnabla \phi \big\|_2^2 
+ \frac{1}{2} X( 2 \Theta_{\phi,\alpha}) \, ,
\end{align}
where $\phi = |\phi| \in \cH_\el$ is normalized and 
$\phi \in H^1(\RR^3)$, $\Theta_{\phi,\alpha}$ is defined in 
\eqref{eq-III-vb.21,1}-\eqref{eq-III-vb.21,2}, and
\begin{align} \label{eq-V-vb.02}
X(A) \ := \ 
\Tr\Big( \sqrt{ |k|^2 + A \:} \: - \: |k| \Big) \, ,
\end{align}
for positive operators $A \geq 0$. Recall that, according to
Theorem~\ref{thm-IV-vb.04} and Corollary~\ref{cor-IV-vb.06}, the
Lieb-Loss energy defined in Eqs.~\eqref{eq-0.03}-\eqref{eq-0.04} is
given by 
\begin{align} \label{eq-V-vb.03,1}
E_\LL(\alpha, \Lambda)
\ = \ 
\inf\big\{ \hcE_{\alpha, \Lambda}(\phi) \; \big| \ 
\phi = |\phi| \in H^1(\RR^3) \, , \ \|\phi\|_2 = 1 \big\} \, .
\end{align}
Ultimately, we compare $\hcE_{\alpha, \Lambda}$ and its infimum 
$E_\LL(\alpha, \Lambda)$ to $\cF_{\beta(\alpha, \Lambda)}$ and its infimum 
$F[\beta(\alpha, \Lambda)]$, respectively, where
\begin{align} \label{eq-V-vb.03,2}
\cF_\beta(\phi) 
\ := \ &
\frac{1}{2} \big\|\vnabla \phi \big\|_2^2 
\: + \: \beta \, \|\phi\|_1 \, ,
\\[1ex] \label{eq-V-vb.03,3}
F[\beta]
\ := \ 
\inf\big\{ \cF_\beta(\phi) \; \big| \ 
\phi = & |\phi| \in H^1(\RR^3) \cap L^1(\RR^3) 
\, , \ \|\phi\|_2 = 1 \big\} \, ,
\\[1ex] \label{eq-V-vb.03,4}
\beta(\alpha, \Lambda)
\ := \ &
\sqrt{\frac{4\alpha}{9\pi}\,} \, \Lambda^3 \, .
\end{align}
In the present section we demonstrate that the minimization in
\eqref{eq-V-vb.03,1} may be restricted to functions supported in the
ball $B(0,L) = \{ x \in \RR^3 : \ |x| < L \}$ of radius $L < \infty$,
provided $L \gg 1$ is sufficiently large. That is, we prove in
Theorem~\ref{thm-V-vb.01} below that 
\begin{align} \label{eq-V-vb.04}
E_\LL^{(L)}(\alpha, \Lambda)
\ := \ 
\inf\Big\{ \hcE_{\alpha, \Lambda}(\phi) \; \Big| \ 
\phi = |\phi| \in Y_L \, , \ \|\phi\|_2 = 1 \big\} \, ,
\end{align}
approximates $E_\LL(\alpha, \Lambda)$, as $L \to \infty$, by showing
that the error made by this restriction is of order $L^{-2}$, as
suggested by the IMS localization formula. Here,
\begin{align} \label{eq-V-vb.04,1}
Y_L \ := \ H^1\big( B(0,L) \big) 
\ \subseteq \
H^1(\RR^3) \cap L^1(\RR^3)
\ \subseteq \
H^1(\RR^3) \ =: \ Y \, ,
\end{align}
and we correspondingly approximate $F[\beta]$ by
\begin{align} \label{eq-V-vb.04,2}
F^{(L)}[\beta]
\ := \ 
\inf\big\{ \cF_\beta(\phi) \; \big| \ & 
\phi = |\phi| \in Y_L \, , \ \|\phi\|_2 = 1 \big\} \, .
\end{align}
%
%
\begin{theorem} \label{thm-V-vb.01}
There exists a universal constant $C < \infty$ such that, 
for all \break $\alpha, \beta, L >0$, $\sigma \geq 0$, and $\Lambda \geq 1$,
\begin{align} \label{eq-V-vb.05,1}
E_\LL^{(L)}(\alpha, \Lambda) - \frac{C}{L^2}
\ \leq \ &
E_\LL(\alpha, \Lambda)
\ \leq \
E_\LL^{(L)}(\alpha, \Lambda) \, ,
\\[1ex] \label{eq-V-vb.05,2}
F^{(L)}[\beta] - \frac{C}{L^2}
\ \leq \ &
F[\beta]
\ \leq \
F^{(L)}[\beta] \, ,
\end{align}
with $E_\LL(\alpha, \Lambda)$, $F(\alpha, \Lambda)$,
$E_\LL^{(L)}(\alpha, \Lambda)$, and $F^{(L)}(\alpha, \Lambda)$ as in
\eqref{eq-V-vb.03,1}, \eqref{eq-V-vb.03,3}, \eqref{eq-V-vb.04}, and
\eqref{eq-V-vb.04,2}, respectively.
\begin{proof} 
The inequalities $E_\LL(\alpha,\Lambda) \leq E_\LL^{(L)}(\alpha,\Lambda)$
and $F[\beta] \leq F^{(L)}[\beta]$ are trivial 
consequences of the inclusions 
$Y_L \subseteq Y$ and $Y_L \subseteq H^1(\RR^3) \cap L^1(\RR^3)$,
respectively.

For the derivation of the lower bound \eqref{eq-V-vb.05,1} on
$E_\LL(\alpha, \Lambda)$ we pick a smooth and compactly supported
function $\eta \in C_0^\infty(\RR^3; \RR_0^+)$, chosen such that
$\supp(\eta) \subseteq B(0,1)$ and $\|\eta\|_2 = 1$. Then we define
\begin{align} \label{eq-V-vb.06,1}
\eta_{L,z}(x) \ := \ L^{-\frac{3}{2}} \, \eta\big[ L^{-1} (x-z)\big] \, ,
\end{align}
for all $L>0$, and we observe that $\| \eta_{L,z} \|_2 = 1$ and 
$\int \eta_{L,z}^2(x) \, d^3z = 1$. We further set
\begin{align} \label{eq-V-vb.06,2}
\rho_L(z) \ := \ \| \eta_{L,z} \, \phi \|_2^2 \, , \quad
\phi_{L,z}(x) \ := \ \left\{ 
\begin{array}{cc}
\frac{\eta_{L,z}(x) \, \phi(x)}{\sqrt{\rho_L(z)\,}} 
& \text{if $\rho_L(z) >0$},
\\ 0 & \text{if $\rho_L(z) =0$},
\end{array} \right.
\end{align}
and observe that $\rho_L$ is a probability density on $\RR^3$. A
variant of the IMS localization formula
\cite{CyconFroeseKirschSimon1987} now yields
\begin{align} \label{eq-V-vb.06,3}
\| \nabla \phi \|_2^2 
\ = \ &
\int \big\| \nabla ( \eta_{L,z} \phi ) \big\|_2^2 \, d^3z 
\: - \: \int |\nabla \eta_{L,z}|^2 \, d^3z 
\nonumber \\[1ex]
\ = \ &
\int  \| \nabla \phi_{L,z} \|_2^2 \; \rho_L(z) \, d^3z 
\: - \: \frac{\|\nabla \eta\|_2^2}{L^2} \, .  
\end{align}
Note that
\begin{align} \label{eq-V-vb.06,4}
\Phi_{\phi,\alpha}^* \, \Phi_{\phi,\alpha}
\ = \ &
\int \big\{ \Phi_{\phi_{L,z},\alpha}^* \, \Phi_{\phi_{L,z},\alpha} \big\} 
\; \rho_L(z) \, d^3z \, ,  
\end{align}
and since $A \mapsto X(A)$ is concave according to
Lemma~\ref{lem-V-vb.03}~(ii), we obtain
\begin{align} \label{eq-V-vb.06.05}
X\big( \Phi_{\phi,\alpha}^* \, \Phi_{\phi,\alpha} \big)
\ \geq \
\int X\big( \Phi_{\phi_{L,z},\alpha}^* \, \Phi_{\phi_{L,z},\alpha} \big) 
\; \rho_L(z) \, d^3z \, . 
\end{align}
Consequently
\begin{align} \label{eq-V-vb.06.06}
\hcE_{\alpha, \Lambda}(\phi) 
\ \geq \ &
\int \hcE_{\alpha, \Lambda}(\phi_{L,z}) \; \rho_L(z) \, d^3z 
\: - \: \frac{\|\nabla \eta\|_2^2}{L^2}
\\[1ex] \nonumber 
\ \geq \ &
\int E_\LL^{(L)}(\alpha, \Lambda) \; \rho_L(z) \, d^3z 
\: - \: \frac{\|\nabla \eta\|_2^2}{L^2}
\ = \ 
E_\LL^{(L)}(\alpha, \Lambda) \: - \: \frac{\|\nabla \eta\|_2^2}{L^2} \, .
\end{align}
Taking the infimum over $\phi \in H^1(\RR^3)$ concludes the proof of
the first inequality in \eqref{eq-V-vb.05,1}. The proof of the first
inequality in \eqref{eq-V-vb.05,2} is similar.
\end{proof}
\end{theorem}
%
For the proof of Theorem~\ref{thm-V-vb.01}, we supply various
properties of $X(A)$ in the following two lemmata. To formulate these
it is convenient to denote
\begin{align} \label{eq-V-vb.07}
K_A \ := \ \sqrt{k^2 + A \, } \, , 
\end{align}
so that
\begin{align} \label{eq-V-vb.08}
X(A) \ = \ \Tr\big[ K_A - K_0 \big] \, .
\end{align}
Since on $\fh$, the multiplication operator 
$\sigma \cdot \bfone \leq |k| \leq \Lambda \cdot \bfone$ is bounded
and bounded invertible, we observe that
\begin{align} \label{eq-V-vb.09}
\sigma \cdot \bfone 
\ \leq \ 
K_0 \ \leq \ K_A 
\ \leq \ 
\big( \Lambda + \|A\|_{\cB[\fh]} \big) \cdot \bfone \, .
\end{align}
%
\begin{lemma} \label{lem-V-vb.02} 
Let $A = A^* \geq 0$ be a bounded self-adjoint operator on $\fh$ 
such that $(k^2 + A)^{\frac{1}{2}}-|k|$ is trace class. Then
\begin{align} \label{eq-V-vb.10}
\Tr\big[ (k^2 + A)^{\frac{1}{2}} - |k| \big]   
\ = \ 
\Tr\big[ A^{\frac{1}{2}} \, \big\{ (k^2 + A)^{\frac{1}{2}} + |k| \big\}^{-1} 
\, A^{\frac{1}{2}} \big] \, .
\end{align}
\begin{proof}
Using \eqref{eq-V-vb.07}-\eqref{eq-V-vb.09}, we have that 
\begin{align} \label{eq-V-vb.11}
\Tr\big[ K_A & - K_0 \big] 
\nonumber \\[1ex]  
\ = \ &
\frac{1}{2} \, \Tr\Big[ \big\{ K_A - K_0 \big\} \, 
\big\{ K_A + K_0 \big\} \, \big\{ K_A + K_0 \big\}^{-1} \Big]
\nonumber \\ &
\: + \: 
\frac{1}{2} \, \Tr\Big[ \big\{ K_A - K_0 \big\} \, 
\big\{ K_A + K_0 \big\}^{-1} \, \big\{ K_A + K_0 \big\} \Big]
\nonumber \\[1ex]  
\ = \ &
\Tr\Big[ \big\{ K_A^2 - K_0^2 \big\} \, 
\big\{ K_A + K_0 \big\}^{-1} \Big] 
\nonumber \\[1ex]  
\ = \ &
\Tr\Big[ A^{\frac{1}{2}} \, \big\{ K_A + K_0 \big\}^{-1} 
\, A^{\frac{1}{2}} \Big] \, ,
\end{align}
where the finiteness of the left side of \eqref{eq-V-vb.11} implies
finiteness of all following lines.
\end{proof} 
\end{lemma}

\begin{lemma} \label{lem-V-vb.03} 
Let $A = A^*, B = B^* \geq 0$ be two bounded self-adjoint operators 
on $\fh$ such that $K_A-K_0$ and $K_B-K_0$ are trace class. Then
\begin{align} \label{eq-V-vb.12}
(\mathrm{i}) \qquad & 
X(A) \ \leq \ X(A+B) \ \leq \ X(A) + X(B) \, , 
\\[1ex] \label{eq-V-vb.13}
(\mathrm{ii}) \qquad & 
A \ \mapsto \ X(A) \ \ \text{is concave} \, .
\end{align}
\begin{proof} Since $A, B \geq 0$ we may use the operator monotonicity
  of the square root and the inverse to infer that $A \mapsto K_A$
and $A \mapsto K_A^{-1}$ are monotone and thus
\begin{align} \label{eq-V-vb.14}
X(A) \ = \ 
\Tr\big[ K_A - K_0 \big]
\ \leq \
\Tr\big[ K_{A+B} - K_0 \big]
\ = \ X(A+B) \, ,
\end{align}
and
\begin{align} \label{eq-V-vb.15}
X(A+B) 
\ = \ &
\Tr\Big[ (A+B) \, \big\{ K_{A+B} + K_0 \big\}^{-1} \Big]
\nonumber \\[1ex] 
\ \leq \ &
\Tr\Big[ A \, \big\{ K_A + K_0 \big\}^{-1} \Big]
\: + \:
\Tr\Big[ B \, \big\{ K_B + K_0 \big\}^{-1} \Big]
\nonumber \\[1ex] 
\ = \ &
X(A) + X(B) \, ,
\end{align}
additionally using Lemma~\ref{lem-V-vb.02}. This yields
\eqref{eq-V-vb.12}.
\\[1ex]
As for \eqref{eq-V-vb.13}, we note that
\begin{align} \label{eq-V-vb.16}
X(\tfrac{1}{2} A + \tfrac{1}{2} B) 
- \tfrac{1}{2} X(A) - \tfrac{1}{2} X(B) 
\ = \ 
\Tr[P-Q] \, ,
\end{align}
where
\begin{align} \label{eq-V-vb.17}
P \ := \ K_{\tfrac{A+B}{2}} 
\quad \text{and} \quad
Q \ := \  \tfrac{1}{2} \, K_A \: + \: \tfrac{1}{2} \, K_B \, .
\end{align}
Now,
\begin{align} \label{eq-V-vb.18}
& \Tr[P-Q] 
\nonumber \\[1ex]
& \ = \  
\frac{1}{2} \, \Tr\big[ (P-Q) \, (P+Q) \, (P+Q)^{-1} \big] \, + \,
\frac{1}{2} \, \Tr\big[ (P-Q) \, (P+Q)^{-1} \, (P+Q) \big] 
\nonumber \\[1ex]
& \ = \
\Tr\big[ (P+Q)^{-1} \, (P^2-Q^2) \big] 
\nonumber \\[1ex] 
& \ = \ 
\frac{1}{2} \, \Tr\Big[ (P+Q)^{-1} \, 
\big\{ k^2 + \tfrac{1}{2} A + \tfrac{1}{2} B
- \tfrac{1}{2} K_A \, K_B - \tfrac{1}{2} K_B \, K_A \big\} \Big] 
\nonumber \\[1ex] 
& \ = \
\frac{1}{4} \, \Tr\Big[ (P+Q)^{-1} \, 
\big\{ K_A - K_B \big\}^2 \Big]
\ \geq \ 0 \, , 
\end{align}
which yields midpoint concavity of $A \mapsto X(A)$, i.e.,
\begin{align} \label{eq-V-vb.19}
X\big( \tfrac{1}{2} A + \tfrac{1}{2} B \big) 
\ \geq \ 
\tfrac{1}{2} X(A) + \tfrac{1}{2} X(B) \, . 
\end{align}
Finally, midpoint concavity and continuity of $A \mapsto X(A)$ implies
general concavity and hence \eqref{eq-V-vb.13}.
\end{proof}
\end{lemma}

\newpage
\section{Upper Bound on $X(2\Theta_{\phi,\alpha})$} 
\label{sec-VI}
%
\setcounter{equation}{0}
We proceed to deriving an upper bound on $E_\LL^{(L)}(\alpha,
\Lambda)$ defined in \eqref{eq-V-vb.04} in terms of
$F^{(L)}(\alpha,\Lambda)$ given in \eqref{eq-V-vb.04,2}. Our
derivation uses two essential tools: 
\begin{compactitem}
\item[(i)] The functional calculus for self-adjoint operators
  described in \cite{AmreinBoutetdeMonvelGeorgescu1996}, which yields
  a good control on projections onto different momentum shells
  emerging from the decomposition $\chi_{\sigma, (1+\eps)\Lambda} =
  \chi_{\sigma, \Lambda} + \chi_{\Lambda, (1+\eps)\Lambda}$, where
  $\eps >0$ and we recall that $\chi_{\sigma, \Lambda} = \bfone[\sigma
  \leq |k| < \Lambda]$. We show that the contribution of
  $\chi_{\Lambda, (1+\eps)\Lambda}$ is negligible, provided $\eps >0$
  is chosen sufficiently small.

\item[(ii)] Inequalities for Schatten-$p$-norms of operators of the
  type ``$f(x)g(-i\nabla)$'', for $1 \leq p \leq 2$, in order to
  estimate the error terms emerging from (i). More specifically,
  Birman and Solomyak have shown \cite{BirmanSolomyak1975,Simon2005}
  that, for any $1 \leq p \leq 2$, there exists a universal constant
  $C_{\mathrm{BS}}(p) < \infty$ such that
\begin{align} \label{eq-VI-vb.00,1}
\| f(x) \: g(i\nabla_x) \|_{\cL^p[\fh_0]}  
\ \leq \
C_{\mathrm{BS}}(p) \, \|f\|_{2;p} \, \|g\|_{2;p} \, , 
\end{align}
provided $\|f\|_{2;p}, \|g\|_{2;p} < \infty$, where
\begin{align} \label{eq-VI-vb.00,2}
\|f\|_{2;p} \ := \ 
\bigg( \sum_{\beta \in \ZZ^3} \| f \cdot \bfone_{Q+\beta} \|_2^p \bigg)^{1/p}
\end{align}
and $Q = [-\frac{1}{2} \, , \, \frac{1}{2}]^3 \subseteq \RR^3$ is
the unit cube centered at the origin.
\end{compactitem}
\begin{theorem} \label{thm-VI-vb.01} 
There exists a universal constant $C < \infty$ such that, 
for all $\alpha, L >0$, all $0 \leq \sigma \leq 1 \leq \Lambda < \infty$, 
all $0 < \eps \leq 1$ and all $\phi = |\phi| \in Y_L$, the estimate 
\begin{align} \label{eq-VI-vb.01}
\frac{1}{2}& X(2\Theta_{\phi,\alpha}) 
\ \leq \ 
\\ \nonumber 
& \sqrt{\frac{4 \alpha}{9 \pi}\,} \, \Big[
\big(\Lambda^3 - \sigma^3 \big) + 54 \eps \, \Lambda^3 +
5 \sigma^{3/2} \, \Lambda^{3/2} \Big] \, \|\phi\|_1
\: + \:
\frac{C \, \alpha^{1/2} (L \, \Lambda +1)^3}{\eps^2 \, L^{3/2} \, \Lambda} \, 
\|\nabla \phi\|_2 \, .
\end{align}
holds true.
\begin{proof} We first apply Lemma~\ref{lem-V-vb.02} and the 
operator monotonicity of $A \mapsto \sqrt{A}$ and $A \mapsto A^{-1}$ 
and observe that 
\begin{align} \label{eq-VI-vb.02}
X(A) \ = \ 
\Tr\Big[ \sqrt{A} \, \big( \sqrt{k^2+A} + |k| \big)^{-1} \, \sqrt{A} \Big] 
\ \leq \ 
\Tr[ \sqrt{A} ] \, .
\end{align}
Secondly, we note that $(\hphi *)^* (\hphi *) = \cF \phi^2 \cF^*$,
where $\cF$ is (componentwise) Fourier transformation. As is
customary, we denote by $\phi(x) := \cF \phi \cF^* \geq 0$ the
corresponding nonnegative multiplication operator, indicating the
change from momentum to position space by explicitly keeping the
argument ``$x$'' for the spatial variable. Using \eqref{eq-VI-vb.02},
the decomposition $\bfone = \chi_{0,\sigma} + \chi_{\sigma,\Lambda} +
\chi_{\Lambda,(1+\varepsilon)\Lambda} +
\bchi_{(1+\varepsilon)\Lambda}$, where $\bchi_r := \bfone - \chi_r$,
and the triangle inequality for the trace norm, we obtain
\begin{align} \label{eq-VI-vb.04}
X( 2\Theta_{\phi,\alpha} )
\ = \ &
X( 2 \Phi_{\phi,\alpha}^* \Phi_{\phi,\alpha})
\ \leq \
\Tr\Big[ \sqrt{2 \Phi_{\phi,\alpha}^* \Phi_{\phi,\alpha} \,} \; \Big]
\ = \ 
\sqrt{2} \, \big\| \Phi_{\phi,\alpha} \|_{\cL^1[\fh]}
\nonumber \\[1ex]
\ = \ &
\frac{(2\alpha)^{1/2}}{(2\pi)^{3/2}} \, 
\big\| \phi(x) \, \chi_{\sigma,\Lambda} \, P_C \big\|_{\cL^1[\fh]}
\\[1ex] \nonumber 
\ \leq \ &
\frac{(2\alpha)^{1/2}}{(2\pi)^{3/2}} \, \Big(
\big\| \chi_{\sigma,\Lambda} \; \phi(x) 
\; \chi_{\sigma,\Lambda} \; P_C \big\|_{\cL^1[\fh]} 
+ 3X_1 + 3X_2 + 3X_3 \Big) \, ,
\end{align}
where 
\begin{align} \label{eq-VI-vb.04,1}
\big\| P_C \; & \chi_{\sigma,\Lambda} \; \phi(x) \; \chi_{\sigma,\Lambda} 
\; P_C \big\|_{\cL^1[\fh]} 
\ = \ 
\big\| \sqrt{\phi(x)} \; \chi_{\sigma,\Lambda} \; P_C \big\|_{\cL^2[\fh]}^2 
\\[1ex] \nonumber 
\ = \ &
2 \, \Big( \Vol[B(0,\Lambda)] - \Vol[B(0,\sigma)] \Big) \,
\bigg( \int \phi(x) \, d^3x \bigg) 
\ = \ 
\frac{8\pi}{3} \big( \Lambda^3 - \sigma^3 \big) \, \|\phi\|_1 
\end{align}
is the main term. Note that the factor $2$ takes into account that
$P_C$ is an orthogonal projection of rank $2$ on $\CC \otimes \RR^3$.
Moreover, we denote by $\Vol[M] := \int \bfone_M(k) \, d^3$ the
three-dimensional Lebesgue measure of a measurable set $M \subseteq
\RR^3$ in \eqref{eq-VI-vb.04,1} and henceforth. Furthermore,
\begin{align} \label{eq-VI-vb.05}
X_1 \ := \ &
\big\| \chi_{0,\sigma} \; \phi(x) \; \chi_{\sigma,\Lambda} \big\|_{\cL^1[\fh_0]} \, ,
\\[1ex] \label{eq-VI-vb.06}
X_2 \ := \ &
\big\| \chi_{\Lambda,(1+\varepsilon)\Lambda} \; \phi(x) 
\; \chi_{\sigma,\Lambda} \big\|_{\cL^1[\fh_0]} \, ,
\\[1ex] \label{eq-VI-vb.07}
X_3 \ := \ &
\big\| \bchi_{(1+\varepsilon)\Lambda} \; \phi(x) 
\; \chi_{\sigma,\Lambda} \big\|_{\cL^1[\fh_0]} 
\end{align}
are error terms we proceed to estimate next. Before we remark that the
Hilbert space in \eqref{eq-VI-vb.05}-\eqref{eq-VI-vb.07} is the space
$\fh_0 := L^2(\RR^3)$ of complex-valued (scalar) square-integrable
functions, as opposed to the one-photon Hilbert space $\fh$ of \break
square-integrable divergence-free vector fields used before. The
factors $3$ on the right side of \eqref{eq-VI-vb.04} account for the
three components of the latter.

Using the trace inequality
$\|AB\|_{\cL^1[\fh_0]} \leq \|A\|_{\cL^2[\fh_0]} \|B\|_{\cL^2[\fh_0]}$ and
$(1+\eps)^3 - 1 \leq 3 \eps (1+\eps)^2 \leq 12 \eps$, we obtain  
\begin{align} \label{eq-VI-vb.08}
X_1 \ \leq \ &
\big\| \chi_{0,\sigma} \; \sqrt{\phi(x)} \big\|_{\cL^2[\fh_0]} \, 
\big\| \sqrt{\phi(x)} \; \chi_{0,\Lambda} \big\|_{\cL^2[\fh_0]} \, 
\ \leq \
\frac{12\pi}{3} \, \Lambda^{3/2} \, \sigma^{3/2} \, \|\phi\|_1 \, ,
\\[1ex] \label{eq-VI-vb.09} 
X_2 \ \leq \ &
\big\| \chi_{\Lambda,(1+\varepsilon)\Lambda} \; \sqrt{\phi(x)} \big\|_{\cL^2[\fh_0]} \, 
\big\| \sqrt{\phi(x)} \; \chi_{0,\Lambda} \big\|_{\cL^2[\fh_0]} \, 
\ \leq \
\frac{144\pi}{3} \, \eps \, \Lambda^3 \, \|\phi\|_1 \, ,
\end{align}
similarly to \eqref{eq-VI-vb.04,1}.

To estimate $X_3$ we pick a smooth function 
$\tg \in C^\infty(\RR; [0,1])$ such that $\tg \equiv 1$ on $\RR_0^-$,
$\tg'\leq 0$, and $\tg \equiv 0 $ on $[1, \infty )$. We then
define a smooth function of compact support by 
\begin{align} \label{eq-VI-vb.10}
g_\eps(\lambda) \ := \ 
\tg\big( \eps^{-1}(\lambda -1) \big) \; 
\tg\big( \eps^{-1}(-\lambda -1) \big) \, .
\end{align}
Note that, for $\eps < 1$ and suitable constants $C_1, C_2, \ldots <\infty$,
we have
\begin{align} \label{eq-VI-vb.11}
\supp(g_\eps) \: \subseteq : (-2,2) \, , & \qquad 
\|g_\eps\|_\infty \: = \: 1 
\\[1ex] \label{eq-VI-vb.12}
\supp\big(g_\eps^{(k)} \big) \: \subseteq \: (-1-\eps, 1) \cup (1, 1+ & \eps)  
\; , \quad 
\big\| g_\eps^{(k)} \big\|_\infty \: \leq \: C_k \, \eps^{-k} \, ,
\end{align}
for all $k \in \NN$. We use the functional calculus developed by
Amrein, Boutet de Monvel, and Georgescu in
\cite[Thm.~6.1.4]{AmreinBoutetdeMonvelGeorgescu1996}. For any
self-adjoint operator $A$ and any $n \in \NN$, this functional
calculus yields the identity
\begin{align} \label{eq-VI-vb.13}
g_\eps(A) \ = \ &
\sum_{k=0}^{n-1} \int_{-\infty}^\infty 
\frac{g_\eps^{(k)}(\lambda) \, d\lambda}{\pi \, k!} 
\: \rIm\big\{i^k ( A-\lambda - i)^{-1} \big\} 
\\[1ex] \nonumber 
& \ \ + 
\int_0^1 \mu^{n-1} \, d\mu \int_{-\infty}^\infty 
\frac{g_\eps^{(n)}(\lambda) \, d\lambda}{\pi \, (n-1)!} 
\ \rIm\big\{i^n ( A-\lambda - i \mu)^{-1} \big\} \, .
\end{align}
We choose $n=3$ and $A:= \Lambda^{-2} k^2 = \Lambda^{-2}
\cF \circ (-\Delta) \circ \cF^* =: \Lambda^{-2} (-\Delta_x)$ and obtain
\begin{align} \label{eq-VI-vb.14}
g_\eps(A) \ = \ &
\int_{-\infty}^\infty 
\big[ g_\eps(\lambda) - \tfrac{1}{2}g_\eps''(\lambda) \big] 
\, \frac{d\lambda}{\pi} \: \rIm\big\{( A-\lambda - i)^{-1} \big\} 
\nonumber \\[1ex] 
& \ \ + 
\int_{-\infty}^\infty g_\eps'(\lambda) \, \frac{d\lambda}{\pi} \: 
\rRe\big\{( A-\lambda - i)^{-1} \big\} 
\\[1ex] \nonumber 
& \ \ -
\int_0^1 \mu^2 \, d\mu \int_{-\infty}^\infty 
g_\eps'''(\lambda) \, \frac{d\lambda}{2 \pi} 
\ \rRe\big\{ ( A-\lambda - i \mu)^{-1} \big\} \, .
\end{align}
We observe that due to the support properties of $g_\eps$ and its
derivatives and the definition of $A = \Lambda^{-2} k^2$, we have 
\begin{align} \label{eq-VI-vb.15}
\chi_{\sigma, \Lambda} 
\ = \ 
g_\eps(A) \, \chi_{\sigma, \Lambda}
\ = \ 
\chi_{\sigma, \Lambda} \, g_\eps(A) 
\, , \quad
g_\eps(A) \, \bchi_{(1+\eps)\Lambda}
\ = \ 
\bchi_{(1+\eps)\Lambda} \, g_\eps(A) 
\ = \ 0 \, ,
\end{align}
which implies that
\begin{align} \label{eq-VI-vb.16}
X_3 \ = \ &
\big\| \bchi_{(1+\varepsilon)\Lambda} \: \phi(x) 
\: \chi_{\sigma,\Lambda} \big\|_{\cL^1[\fh_0]} 
\ = \ 
\big\| \bchi_{(1+\varepsilon)\Lambda} \: \phi(x) 
\: g_\eps(A) \: \chi_{\sigma,\Lambda} \big\|_{\cL^1[\fh_0]} 
\nonumber \\[1ex]
\ = \ &
\big\| \bchi_{(1+\varepsilon)\Lambda} \: \big[ g_\eps(A) \, , \, \phi(x) \big] 
\: \chi_{\sigma,\Lambda} \big\|_{\cL^1[\fh_0]} 
\\[1ex] \nonumber 
\ \leq \ &
\int_{-\infty}^\infty \frac{d\lambda}{\pi} \: 
\big( |g_\eps(\lambda)| + |g_\eps'(\lambda)| + |g_\eps''(\lambda)| \big) \ 
\Big\| \big[ R(\lambda+i) \, , \, \phi(x) \big] 
\: \chi_{\sigma,\Lambda} \Big\|_{\cL^1[\fh_0]}  
\\[1ex] \nonumber 
& \ \ 
+ \int_0^1 \mu^2 \, d\mu \: \int_{-\infty}^\infty \frac{d\lambda}{2 \pi} 
\: |g_\eps'''(\lambda)| \ \Big\| \big[ R(\lambda+i\mu) \, , \, \phi(x) \big] 
\: \chi_{\sigma,\Lambda} \Big\|_{\cL^1[\fh_0]}  \, ,
\end{align}
where $R(z):= \big( - \Lambda^{-2} \Delta_x - z \big)^{-1}$, with
$z \in \CC \setminus \RR$. Now, note that
\begin{align} \label{eq-VI-vb.17}
\big[ R(z) \, , \, \phi(x) \big]
\ = \ &
\Lambda^{-2} \, R(z) \, [ \Delta_x , \phi(x)] \, R(z) 
\\[1ex] \nonumber
\ = \ &
\Lambda^{-2} \, R(z) \, 
\big( \nabla_x \cdot \nabla\phi(x) + \nabla\phi(x) \cdot \nabla_x \big) 
\, R(z) \, ,
\end{align}
and hence
\begin{align} \label{eq-VI-vb.18}
\Big\| \big[ R(\lambda+i\mu) & \, , \, \phi(x) \big] 
\: \chi_{\sigma,\Lambda} \Big\|_{\cL^1[\fh_0]} 
\nonumber \\[1ex] 
\ \leq \ &
\frac{2}{\Lambda^2} \, 
\| \nabla_x \, R(\lambda+i\mu) \|_{\cB[\fh_0]} \, 
\| R(\lambda+i\mu) \|_{\cB[\fh_0]} \,
\| \nabla\phi(x) \: \chi_{\sigma,\Lambda} \|_{\cL^1[\fh_0]}  
\nonumber \\[1ex] 
\ \leq \ &
\frac{4}{\mu^2 \, \Lambda} \, 
\| \nabla\phi(x) \: \chi_{\sigma,\Lambda} \|_{\cL^1[\fh_0]}  \, ,
\end{align}
using that, for all $\lambda \in [-2,2]$ and $\mu \in (0,1)$, 
\begin{align} \label{eq-VI-vb.19}
\| R(\lambda+i\mu) \|_{\cB[\fh_0]} 
\ = \ &
\sup_{r>0}\Big\{ \big| (r/\Lambda)^2 - \lambda - i\mu \big|^{-1} \Big\}
\ = \ 
\frac{1}{\mu} \, ,
\\[1ex] \label{eq-VI-vb.20}
\| \nabla_x \, R(\lambda+i\mu) \|_{\cB[\fh_0]} 
\ = \ &
\Lambda \, \sup_{r>0}\Big\{ (r/\Lambda) 
\big| (r/\Lambda)^2 - \lambda - i\mu \big|^{-1} \Big\}
\ \leq \ 
\frac{2 \Lambda}{\mu} \, .
\end{align}
Inserting \eqref{eq-VI-vb.18} into \eqref{eq-VI-vb.16} and
additionally taking \eqref{eq-VI-vb.11}-\eqref{eq-VI-vb.12},
as well as $\eps \in (0,1)$ into account, we arrive at
\begin{align} \label{eq-VI-vb.21}
X_3 \ \leq \ 
\frac{C}{\eps^2 \, \Lambda} \:
\| \nabla\phi(x) \: \chi_{\sigma,\Lambda} \|_{\cL^1[\fh_0]}  \, ,
\end{align}
for some universal constant $C <\infty$. To estimate the trace norm on
the right side of \eqref{eq-VI-vb.21} we first conjugate the operators
by a suitable unitary dilatation, which implements the change of
length scale $(x, k) \mapsto (Lx, k/L)$ and does not change the norm,
and then apply Inequality~\eqref{eq-VI-vb.00,1} with $p =1$. These
steps lead us to 
\begin{align} \label{eq-VI-vb.23}
\| \nabla\phi(x) \: & \chi_{\sigma,\Lambda}(k) \|_{\cL^1[\fh_0]}  
\ = \ 
\| \nabla\phi(Lx) \: \chi_{\sigma,\Lambda}(k/L) \|_{\cL^1[\fh_0]}  
\\[1ex] \nonumber 
\ = \ &
\| \nabla\phi(Lx) \: \chi_{L\sigma,L\Lambda}(k) \|_{\cL^1[\fh_0]}  
\ \leq \ 
C_{\mathrm{BS}}(1) \, \| \nabla\phi \|_{2;1} \: 
\| \chi_{L\sigma,L\Lambda} \|_{2;1}  
\\[1ex] \nonumber 
\ = \ &
C_{\mathrm{BS}}(1) \,
\| \nabla\phi(Lx) \|_2 \: \sum_{\gamma \in \ZZ^3} 
\sqrt{ \Vol\big[ B(0,L\Lambda) \cap (Q+\gamma) \big] \, } \, ,
\end{align}
where we use that $\nabla\phi$ is supported in $B(0,L)$, hence
$x \mapsto \nabla\phi(Lx)$ is supported in $B(0,1) \subseteq Q$,
and in the sum 
$\sum_{\beta \in \ZZ^3} \|\nabla\phi(Lx) \, \bfone_{Q+\beta}\|_2$ only
the term corresponding to $\beta = 0$ contributes.
Now, $\Vol[B(0,L\Lambda) \cap (Q+\gamma)] \leq \Vol[(Q+\gamma)] = 1$
and $\Vol[B(0,L\Lambda) \cap (Q+\gamma)] \leq \Vol[(Q+\gamma)] = 0$
unless $|\gamma| \leq L\Lambda + \sqrt{3}$ which implies that
\begin{align} \label{eq-VI-vb.24}
\sum_{\gamma \in \ZZ^3} \sqrt{ \Vol\big[ B(0,L\Lambda) \cap (Q+\gamma) \big] \, }
\ \leq \ &
\sum_{\gamma \in \ZZ^3} \bfone_{B(0,L\Lambda + \sqrt{3})}(\gamma)
\\[1ex] \nonumber
\ \leq \ 
\Vol\big[ B\big(0,L\Lambda + \tfrac{3}{2}\sqrt{3} \big) & \big] 
\ \leq \
\frac{4\pi}{3} \, (L \, \Lambda + 3)^3 \, ,
\end{align}
using that $\tfrac{3}{2} \sqrt{3} \leq 3$. Furthermore, 
$\| \nabla\phi(Lx) \|_2 = L^{-3/2} \| \nabla\phi(x) \|_2$, and thus
\begin{align} \label{eq-VI-vb.25}
\| \nabla\phi(x) \: & \chi_{\sigma,\Lambda}(k) \|_{\cL^1[\fh_0]}  
\ \leq \ 
\frac{4\pi \, C_{\mathrm{BS}}}{3 \, L^{3/2}} \, 
(L \, \Lambda + 3)^3 \: \| \nabla\phi \|_2 \, .
\end{align}
Inserting this into \eqref{eq-VI-vb.21} we finally obtain
\begin{align} \label{eq-VI-vb.26}
X_3 \ \leq \ 
\frac{C \, (L \, \Lambda + 1)^3}{\eps^2 \, L^{3/2} \, \Lambda} 
\: \| \nabla\phi \|_2 \, ,
\end{align}
for a suitable constant $C < \infty$.
Estimate~\eqref{eq-VI-vb.01} now follows from inserting
\eqref{eq-VI-vb.04,1}, \eqref{eq-VI-vb.08}, \eqref{eq-VI-vb.09}, and
\eqref{eq-VI-vb.26} into \eqref{eq-VI-vb.04,1}.
\end{proof}
\end{theorem}

\newpage
\section{Lower Bound on $X(\Theta_{\phi,\alpha})$} 
\label{sec-VII}
%
\setcounter{equation}{0}
In order to complement the upper bound on $X(\Theta_{\phi,\alpha})$ from
Section~\ref{sec-VI} by a corresponding lower bound we first derive a
general inequality on $X(A)$ of the form $X(A) \geq
\Tr\big[A^{1/2}\big] - 2 \Lambda^{1-p} \: \Tr\big[ A^{p/2} \big]$,
where $p$ is any exponent between $\frac{1}{2}$ and $1$. By another
application of the Birman-Solomyak inequality \eqref{eq-VI-vb.00,1} we
then estimate the emerging error term by a multiple of $\|\phi\|_1^p$.

We begin by deriving a general lower bound on $X(A)$ only using that
$|k| \leq \Lambda \cdot \bfone$ on $\fh$.
\begin{lemma} \label{lem-VII-vb.01}
Let $A \geq 0$ be a nonnegative self-adjoint operator on
$\fh$ such that $A^{1/2} \in \cL[\fh]$ is trace-class and assume
that $0 < p < 1$. Then
\begin{align} \label{eq-VII-vb.01}
X(A) \ \geq \ 
\Tr\big[ A^{1/2} \big] \: - \: 2 \Lambda^{1-p} \: \Tr\big[ A^{p/2} \big] \, .
\end{align}
\begin{proof} We recall from Lemma~\ref{lem-V-vb.02} that
\begin{align} \label{eq-VII-vb.02}
X(A) \ = \ &
\Tr\big[ A^{1/2} \, \big\{ (k^2 + A)^{1/2} + |k| \big\}^{-1} 
\, A^{1/2} \big]  
\nonumber \\[1ex]
\ = \ &
\Tr\big[ A^{1/2} \, \big\{ K_A + K_0 \big\}^{-1} 
\, A^{1/2} \big] \, ,
\end{align}
with $K_A := \sqrt{k^2 + A}$. From the second resolvent identity 
we derive 
\begin{align} \label{eq-VII-vb.03}
\frac{1}{K_A+K_0} 
\ = \ &
\frac{1}{K_A} \: - \: \frac{1}{K_A} \, K_0 \, \frac{1}{K_A+K_0} 
\nonumber \\[1ex] 
\ = \ &
\frac{1}{K_A} \: - \: 
\frac{1}{K_A} \, K_0 \, \frac{1}{K_A} \: + \: 
\frac{1}{K_A} \, K_0 \, \frac{1}{K_A+K_0} \, K_0 \, \frac{1}{K_A} 
\nonumber \\[1ex] 
\ \geq \ &
\frac{1}{K_A} \: - \: 
\frac{1}{K_A} \, K_0 \, \frac{1}{K_A} \, ,
\end{align}
which implies
\begin{align} \label{eq-VII-vb.04}
X(A) 
\ \geq \ 
\Tr\big[ A^{1/2} \, K_A^{-1} \, A^{1/2} \big] \: - \: 
\Tr\big[ A^{1/2} \, K_A^{-1} \, K_0 \, K_A^{-1} \, A^{1/2} \big] \, .
\end{align}
Since
\begin{align} \label{eq-VII-vb.05}
K_0 \ = \ |k| \ \leq \ \Lambda^{1-p} |k|^p \ \leq \ \Lambda^{1-p} K_A^{p/2} \, ,
\end{align}
and $K_A \geq A^{1/2}$, we have that
\begin{align} \label{eq-VII-vb.06}
\Tr\big[ A^{1/2} \, K_A^{-1} \, K_0 \, K_A^{-1} \, A^{1/2} \big] 
\ \leq \ &
\Lambda^{1-p} \, \Tr\big[ A^{1/2} \, K_A^{-2+(p/2)} \, A^{1/2} \big] 
\nonumber \\[1ex]
\ \leq \ &
\Lambda^{1-p} \, \Tr\big[ A^{p/2} \big] \, .
\end{align}
By operator monotonicity we further have 
\begin{align} \label{eq-VII-vb.07}
A^{1/2} \: & - \: A^{1/2} \, K_A^{-1} \, A^{1/2} 
\ \leq \ 
A^{1/2} \: - \: \frac{A}{\sqrt{\Lambda^2 +A\,}}
\nonumber \\[1ex]
\ = \ &
\frac{A^{1/2}}{\sqrt{\Lambda^2 +A \,}}
\Big( \sqrt{\Lambda^2 +A \,} \, - \, A^{1/2} \Big)
\ = \ 
\frac{A^{1/2} \, \Lambda^2}{
\sqrt{\Lambda^2 +A \,} \, \big( \sqrt{\Lambda^2 +A \,} + A^{1/2} \big) }
\nonumber \\[1ex]
\ \leq \ &
\frac{\Lambda^2 \, A^{1/2}}{\Lambda^2 +A} 
\ \leq \ 
\frac{\Lambda^2 \, A^{1/2}}{\Lambda^{1+p} \, A^{(1-p)/2}} 
\ = \ 
\Lambda^{1-p} \, A^{p/2} \, .
\end{align}
Inserting \eqref{eq-VII-vb.06} and \eqref{eq-VII-vb.07} into
\eqref{eq-VII-vb.04}, we arrive at the claim.
\end{proof}
\end{lemma}
As described above, we now use Lemma~\ref{lem-VII-vb.01} to derive a
lower bound on $X(2\Theta_{\phi,\alpha})$.
\begin{theorem} \label{thm-VII-vb.02} 
There exists a universal constant $C < \infty$ such that, 
for all $\alpha, L >0$, all $0 \leq \sigma \leq 1 \leq \Lambda < \infty$, 
all $0 < \eps \leq 1$ and all $\phi = |\phi| \in Y_L$, the estimate 
\begin{align} \label{eq-VII-vb.08}
\frac{1}{2} X(2\Theta_{\phi,\alpha}) 
\ \geq \ 
\sqrt{\frac{4 \alpha}{9 \pi}\,} \, 
\big(\Lambda^3 - \sigma^3 \big) \, \|\phi\|_1
\: - \:
\frac{C \, \alpha^{1/4} \, \Lambda^{1/2} (L \, \Lambda +1)^3}{L^{3/2}} \, 
\sqrt{ \|\phi\|_1 \, } 
\end{align}
holds true.
\begin{proof} We first use that $A \mapsto X(A)$ is monotonically
increasing. Since 
\begin{align} \label{eq-VII-vb.09}
\Theta_{\phi,\alpha} \ = \ & 
\Phi_{\phi,\alpha}^* \Phi_{\phi,\alpha} 
\ = \ 
\frac{\alpha}{(2\pi)^3} \, 
P_C \, \chi_{\sigma,\Lambda} \, \phi(x)^2 \, \chi_{\sigma,\Lambda} \, P_C
\\[1ex] \nonumber
\ \geq \ &
\frac{\alpha}{(2\pi)^3} \, 
P_C \, \chi_{\sigma,\Lambda} \, \phi(x) \, \chi_{\sigma,\Lambda} \, 
P_C \, \chi_{\sigma,\Lambda} \, \phi(x) \, \chi_{\sigma,\Lambda} \, P_C
\\[1ex] \nonumber
\ = \ &
\big[ \alpha^{1/2} \, (2\pi)^{-3/2} \; 
P_C \, \chi_{\sigma,\Lambda} \, 
\phi(x) \, \chi_{\sigma,\Lambda} \, P_C \big]^2 \, ,
\end{align}
we obtain from Lemma~\ref{lem-VII-vb.01} with $p = \frac{1}{2}$ that
\begin{align} \label{eq-VII-vb.10}
\frac{1}{2} X(2\Theta_{\phi,\alpha}) 
\ \geq \ &
\frac{1}{2} X\Big( \big[ (2\alpha)^{1/2} \, (2\pi)^{-3/2} \; 
P_C \, \chi_{\sigma,\Lambda} \, 
\phi(x) \, \chi_{\sigma,\Lambda} \, P_C \big]^2 \Big) 
\nonumber \\[1ex] 
\ \geq \ &
\frac{(2\alpha)^{1/2}}{2 (2\pi)^{3/2}} 
\, \Tr\big( P_C \, \chi_{\sigma,\Lambda} \, 
\phi(x) \, \chi_{\sigma,\Lambda} \, P_C \big) 
\nonumber \\ 
\  \ & \quad 
\: - \: \frac{(2\alpha)^{1/4} \Lambda^{1/2}}{(2\pi)^{3/4}} 
\: \Tr\big( [P_C \, \chi_{\sigma,\Lambda} \, 
\phi(x) \, \chi_{\sigma,\Lambda} \, P_C]^{1/2} \big) \, .
\\[1ex] \nonumber 
\ \geq \ &
\sqrt{\frac{4 \alpha}{9 \pi}\,} \, \big(\Lambda^3 - \sigma^3 \big) 
\, \|\phi\|_1
\: - \: \frac{ (2\alpha)^{1/4} \Lambda^{1/2}}{(2\pi)^{3/4}} \,
\big\| \sqrt{\phi(x)} \, \chi_{\sigma,\Lambda}(k) \big\|_{\cL^1[\fh_0]} \, .
\end{align}
To estimate the second term on the right side of \eqref{eq-VII-vb.10}
we proceed as in \eqref{eq-VI-vb.21}-\eqref{eq-VI-vb.25}. After
unitary rescaling $(x,k) \mapsto (Lx,k/L)$, we apply
\eqref{eq-VI-vb.00,1} again and get
\begin{align} \label{eq-VII-vb.11}
\big\| \sqrt{\phi(x)} \, \chi_{\sigma,\Lambda}(k) \big\|_{\cL^1[\fh_0]}
\ = \ &
\big\| \sqrt{\phi(Lx)} \, \chi_{L\sigma,L\Lambda}(k) \big\|_{\cL^1[\fh_0]}
\\[1ex] \nonumber
\ \leq \ & 
C_{\mathrm{BS}}(1) \, 
\big\| \sqrt{\phi(Lx)} \big\|_{2;1} \: 
\big\| \chi_{L\sigma,L\Lambda}(k) \big\|_{2;1} 
\\[1ex] \nonumber
\ \leq \ & 
\frac{4 \pi \, C_{\mathrm{BS}}(1)}{3} \, (L\Lambda+3)^3 \, 
\big\| \sqrt{\phi(Lx)} \big\|_{2;1} \, , 
\end{align}
where the last estimate results from
\eqref{eq-VI-vb.23}-\eqref{eq-VI-vb.24}. Since $x \mapsto \phi(Lx)$ is
supported in $B(0,1) \subseteq Q = [-\frac{1}{2}, \frac{1}{2}]^3$, we
further have
\begin{align} \label{eq-VII-vb.13}
\big\| \sqrt{\phi(Lx)} \big\|_{2;1}
\ = \ &
\sum_{\beta \in \ZZ^3} 
\big\| \sqrt{\phi(Lx)} \cdot \bfone_{Q+\beta} \big\|_2
\ = \ 
\big\| \sqrt{\phi(Lx)} \big\|_2
\nonumber \\[1ex]
\ = \ &
\| \phi(Lx) \|_1^{1/2}
\ = \ 
L^{-3/2} \,  \| \phi \|_1^{1/2} \, .
\end{align}
Finally, inserting \eqref{eq-VII-vb.13} into
\eqref{eq-VII-vb.11}, we arrive at \eqref{eq-VII-vb.08}.
\end{proof}
\end{theorem}

\newpage
\section{Asymptotics of the Lieb-Loss Energy} 
\label{sec-VIII}
%
\setcounter{equation}{0}
We turn to the proof of the main result of this paper,
Theorem~\ref{thm-0.1}, stated below again for the reader's
convenience. In our proof a key role is played by the scaling relation
the effective energy $F[\beta]$ obeys. $F[\beta]$ is defined in
\eqref{eq-V-vb.03,3} as the infimum of the functional $\cF_\beta > 0$
over $L^2$-normalized functions in $H^1(\RR^3) \cap L^1(\RR^3)$. In
\cite{Hach2020} one of us showed that this infimum is attained for
some $\phi_\beta$ and hence is actually a minimum. A major issue in
this regard is non-reflexivity of the
$L^1$-space precluding a naive application of the direct method of the
calculus of variations. This was remedied by using the theory of
uniform convex spaces and the Milman-Pettis theorem. Subsequently an
explicit characterization of the minimizer (up to spherical
rearrangement) can be given in terms of a Bessel function. In
particular,
\begin{align} \label{eq-VIII-vb.01}
F[1] \ > \ 0 
\end{align}
is a positive constant, and it is then not difficult to see that
$F[\beta]$ scales as
\begin{align} \label{eq-VIII-vb.02}
F[\beta] \ = \ \beta^{4/7} \, F[1] \, , 
\end{align}
for all $\beta >0$.
\begin{theorem} \label{thm-VIII-vb.1}
There exists a universal constant $C < \infty$ such that, 
for all $\alpha>0$ and $\Lambda \geq 1$, the estimate
\begin{align} \label{eq-VIII-vb.03}
- C \, \alpha^{\frac{4}{49}} \Lambda^{-\frac{4}{49}}
\ \leq \ 
\frac{E_\LL(\alpha, \Lambda)}{F_1 \, \alpha^{2/7} \, \Lambda^{12/7}} 
\: - \: 1 
\ \leq \  
C \, \alpha^{\frac{4}{105}} \Lambda^{-\frac{4}{105}} 
\end{align}
holds true.
\begin{proof} We first take the infrared limit $\sigma \to 0$. Note that
$E_\LL$, $E_\LL^{(L)}$, $F$, $F^{(L)}$, $X(2\Theta_{\phi,\alpha})$,
and all error terms are continuous at $\sigma = 0$, and we can simply
set $\sigma := 0$ everywhere. Then Theorems~\ref{thm-VI-vb.01} and
\ref{thm-VII-vb.02} with $p = 1/2$ yield
\begin{align} \label{eq-VIII-vb.05}
\frac{1}{2} X(2\Theta_{\phi,\alpha}) 
\: - \: 
\sqrt{\frac{4 \alpha}{9 \pi}\,} \, \Lambda^3 \, \|\phi\|_1
\ \leq \ &
C \, \eps \, \Lambda^3 \, \|\phi\|_1 \, + \, 
C \, \alpha^{1/2} \, \eps^{-2} \, L^{3/2} \, \Lambda^2 \, \|\nabla \phi\|_2 \, ,
\\[1ex] \label{eq-VIII-vb.06}
\frac{1}{2} X(2\Theta_{\phi,\alpha}) 
\: - \: 
\sqrt{\frac{4 \alpha}{9 \pi}\,} \, \Lambda^3 \, \|\phi\|_1
\ \geq \ &
- C \, \alpha^{1/4} \, L^{3/2} \, \Lambda^{7/2} \, \|\phi\|_1^{1/2} \, ,
\end{align}
some constant $C_1 < \infty$ and any $\phi = |\phi| \in Y_L$
with $\|\phi\|_2 = 1$, provided that $L \geq \Lambda^{-1}$.

We first derive the upper bound in \eqref{eq-VIII-vb.03}. 
From \eqref{eq-VIII-vb.05} we obtain
\begin{align} \label{eq-VIII-vb.07}
\hcE_{\alpha,\Lambda}(\phi) 
\ = \ & 
\frac{1}{2} \|\nabla \phi\|_2^2  + \frac{1}{2} X(2\Theta_{\phi,\alpha}) 
\nonumber \\[1ex] 
\ \leq \ &
\frac{1}{2} (1 + \delta) \, \|\nabla \phi\|_2^2
\: + \:
\sqrt{\frac{4 \alpha}{9 \pi}\,} \, \Lambda^3 \, (1 + C_2 \eps) \, \|\phi\|_1
\: + \:
C_2 \, \alpha \, \delta^{-1} \, \eps^{-4} \, L^3 \, \Lambda^4 
\nonumber \\[1ex] 
\ \leq \ &
(1 + \delta) \, \cF_{\beta_2}(\phi) 
\: + \:
C_2 \, \alpha \, \delta^{-1} \, \eps^{-4} \, L^3 \, \Lambda^4 \, ,
\end{align}
where $\cF_\beta$ is defined in \eqref{eq-V-vb.03,2} and
\begin{align} \label{eq-VIII-vb.08}
\beta_2 \ := \ \beta_0 \, \frac{1 + C_2 \eps}{1 + \delta} 
\; , \quad
\beta_0 \ \equiv \ \beta(\alpha,\Lambda) \ = \ 
\sqrt{\frac{4 \alpha}{9 \pi}\,} \, \Lambda^3 \, ,
\end{align}
for some $C_2 < \infty$ and all $0 < \delta \leq 1$.
Taking the infimum over all $\phi = |\phi| \in Y_L$ with $\|\phi\|_2 = 1$
in \eqref{eq-VIII-vb.07}, we further have
\begin{align} \label{eq-VIII-vb.09}
E_\LL^{(L)}(\alpha,\Lambda)
\ \leq \ & 
(1 + \delta) \, F^{(L)}[\beta_2]
\: + \:
C_2 \, \alpha \, \delta^{-1} \, \eps^{-4} \, L^3 \, \Lambda^4 \, .
\end{align}
The localization estimates \eqref{eq-V-vb.05,1}-\eqref{eq-V-vb.05,2}
now imply 
\begin{align} \label{eq-VIII-vb.09,1}
E_\LL(\alpha,\Lambda)
\ \leq \ & 
(1 + \delta) \, F[\beta_2]
\: + \:
C_2 \, \alpha \, \delta^{-1} \, \eps^{-4} \, L^3 \, \Lambda^4 
\: + \:
C_3 \, L^{-2} \, ,
\end{align}
for some constant $C_3 < \infty$. From the scaling relation 
\eqref{eq-VIII-vb.02}, we get
\begin{align} \label{eq-VIII-vb.10}
(1 + \delta) \, F[\beta_2]
\ = \ &
(1 + \delta) \, \beta_2^{4/7} \, F[1]
\ = \ 
(1 + C_2 \eps)^{4/7} \, (1 + \delta)^{3/7} \, \beta_0^{4/7} \, F[1]
\nonumber \\[1ex]
\ \leq \ &
(1 + C_4 \eps + C_4 \delta) \, F[\beta_0] \, ,
\end{align}
for some $C_4 < \infty$, and inserting this into
\eqref{eq-VIII-vb.09,1}, we arrive at the intermediate estimate,
stating that there exists a universal constant $C_5 < \infty$,
such that
\begin{align} \label{eq-VIII-vb.11}
& \frac{E_\LL(\alpha,\Lambda)}{F[\beta(\alpha,\Lambda)]} \, - \, 1
\\[1ex] \nonumber
& \ \leq \ 
C_5 \Big( \eps \: + \: \delta \: + \: 
\alpha^{-2/7} \, L^{-2} \, \Lambda^{-12/7} \: + \: 
\alpha^{5/7} \, \delta^{-1} \, \eps^{-4} \, L^3 \, \Lambda^{16/7} \Big) 
\end{align}
holds for all $\eps, \delta \in (0,1]$, $\alpha >0$, $\Lambda \geq 1$,
and $L > \Lambda^{-1}$. As $\alpha$ enters the right side of 
\eqref{eq-VIII-vb.11} only in negative powers, we may assume 
$\alpha \in (0,1]$ w.l.o.g. To meet these requirements, we set 
\begin{align} \label{eq-VIII-vb.12}
\eps \; := \; \delta \; := \; \alpha^r \Lambda^{-s} 
\quad \text{and} \quad
L \; := \; \alpha^{-t} \Lambda^{u-1} \; , \quad
\end{align}
for $r,s,t,u \geq 0$ to be chosen later. Then
\begin{align} \label{eq-VIII-vb.13}
\eps + \delta + & \alpha^{-\frac{2}{7}} L^{-2} \Lambda^{-\frac{12}{7}} +  
\alpha^{\frac{5}{7}} \delta^{-1} \eps^{-4} L^3 \Lambda^{\frac{16}{7}}
\\[1ex] \nonumber
\ = \ &
2 \alpha^r \Lambda^{-s} 
+ \alpha^{2t-\frac{2}{7}} \Lambda^{\frac{2}{7}-2u} 
+ \alpha^{\frac{5}{7}-5r-3t} \Lambda^{5s+3u-\frac{5}{7}} 
\ \leq \
4 \, \alpha^{a/7} \, \Lambda^{-b/7} \, ,
\end{align}
with
\begin{align} \label{eq-VIII-vb.14}
a \ := \ &
\min\big\{ 7r, \; 14t-2, \; 5-35r-21t \big\} \, ,
\\[1ex] \label{eq-VIII-vb.15}
b \ := \ &
\min\big\{ 7s, \; 14u-2, \; 5-35s-21u \big\} \, .
\end{align}
We choose $r,s,t,u$ so that all three terms in both
\eqref{eq-VIII-vb.14} and \eqref{eq-VIII-vb.15} are equal, i.e.,
$r := s := 4/105$ and $t := u := 17/105$. This yields 
$a = b = 4/15$ and hence the upper bound
\begin{align} \label{eq-VIII-vb.16}
\frac{E_\LL(\alpha,\Lambda)}{F[\beta(\alpha,\Lambda)]} \, - \, 1
\ \leq \ 
4C_5 \; \alpha^{\frac{4}{105}} \; \Lambda^{-\frac{4}{105}} 
\end{align}
in \eqref{eq-VIII-vb.03}.

We similarly proceed for the lower bound in \eqref{eq-VIII-vb.03}. 
From \eqref{eq-VIII-vb.06} we obtain
\begin{align} \label{eq-VIII-vb.17}
\hcE_{\alpha,\Lambda}(\phi) 
\ = \ & 
\frac{1}{2} \|\nabla \phi\|_2^2  + \frac{1}{2} X(2\Theta_{\phi,\alpha}) 
\nonumber \\[1ex] 
\ \geq \ &
\frac{1}{2} \|\nabla \phi\|_2^2
\: + \:
\sqrt{\frac{4 \alpha}{9 \pi}\,} \, \Lambda^3 \, 
(1 - \delta) \, \|\phi\|_1
\: - \: C_6 \, \delta^{-1} \, L^3 \, \Lambda 
\nonumber \\[1ex] 
\ = \ &
\cF_{\beta_3}(\phi) 
\: - \: C_6 \, \delta^{-1} \, L^3 \, \Lambda \, ,
\end{align}
for some $C_2 < \infty$ and all $0 < \delta \leq 1$, where 
\begin{align} \label{eq-VIII-vb.18}
\beta_3 \ := \ \beta_0 \, (1 - \delta)
\; , \quad
\beta_0 \ \equiv \ \beta(\alpha,\Lambda) \ = \ 
\sqrt{\frac{4 \alpha}{9 \pi}\,} \, \Lambda^3 \, .
\end{align}
Taking the infimum over all $\phi = |\phi| \in Y_L$ with $\|\phi\|_2 = 1$
in \eqref{eq-VIII-vb.17}, we further have
\begin{align} \label{eq-VIII-vb.19}
E_\LL^{(L)}(\alpha,\Lambda)
\ \geq \ & 
F^{(L)}[\beta_3] \: - \: C_6 \, \delta^{-1} \, L^3 \, \Lambda \, .
\end{align}
The localization estimates \eqref{eq-V-vb.05,1}-\eqref{eq-V-vb.05,2}
now imply 
\begin{align} \label{eq-VIII-vb.20}
E_\LL(\alpha,\Lambda)
\ \geq \ & 
F[\beta_3] \: - \: C_7 \, L^{-2} 
\: - \: C_6 \, \delta^{-1} \, L^3 \, \Lambda^4 \, ,
\end{align}
for some constant $C_7 < \infty$. Again invoking the scaling relation 
\eqref{eq-VIII-vb.02}, we get
\begin{align} \label{eq-VIII-vb.21}
F[\beta_3]
\ = \ 
\beta_3^{4/7} \, F[1]
\ = \ 
(1 - \delta)^{4/7} \, \beta_0^{4/7} \, F[1]
\ \geq \ 
(1 - \delta) \, F[\beta_0] \, ,
\end{align}
and thus there exists a constant $C_8 < \infty$ such that
\begin{align} \label{eq-VIII-vb.22}
\frac{E_\LL(\alpha,\Lambda)}{F[\beta(\alpha,\Lambda)]} \, - \, 1
\ \geq \ 
- C_8 \Big( \delta \: + \: 
\alpha^{-2/7} \, L^{-2} \, \Lambda^{-12/7} \: + \: 
\alpha^{-2/7} \, \delta^{-1} \, L^3 \, \Lambda^{16/7} \Big) 
\end{align}
holds for all $\delta, \alpha \in (0,1]$, $\Lambda \geq 1$,
and $L > \Lambda^{-1}$. Again we set 
\begin{align} \label{eq-VIII-vb.23}
\delta \; := \; \alpha^r \Lambda^{-s} 
\quad \text{and} \quad
L \; := \; \alpha^{-t} \Lambda^{u-1} \; , \quad
\end{align}
for $r,s,t,u \geq 0$ to be chosen later and obtain
\begin{align} \label{eq-VIII-vb.24}
\delta + & \alpha^{-\frac{2}{7}} L^{-2} \Lambda^{-\frac{12}{7}} +  
\alpha^{-\frac{2}{7}} \delta^{-1} L^3 \Lambda^{\frac{16}{7}}
\\[1ex] \nonumber
\ = \ &
\alpha^r \Lambda^{-s} 
+ \alpha^{2t-\frac{2}{7}} \Lambda^{\frac{2}{7}-2u} 
+ \alpha^{\frac{5}{7}-r-3t} \Lambda^{s+3u-\frac{5}{7}} 
\ \leq \
3 \, \alpha^{a/7} \, \Lambda^{-b/7} \, ,
\end{align}
with
\begin{align} \label{eq-VIII-vb.25}
a \ := \ &
\min\big\{ 7r, \; 14t-2, \; 5-7r-21t \big\} \, ,
\\[1ex] \label{eq-VIII-vb.26}
b \ := \ &
\min\big\{ 7s, \; 14u-2, \; 5-7s-21u \big\} \, .
\end{align}
We choose $r,s,t,u$ so that all three terms in both
\eqref{eq-VIII-vb.14} and \eqref{eq-VIII-vb.15} are equal, i.e.,
$r := s := 4/49$ and $t := u := 9/49$. This yields 
$a = b = 4/7$ and hence the lower bound
\begin{align} \label{eq-VIII-vb.27}
\frac{E_\LL(\alpha,\Lambda)}{F[\beta(\alpha,\Lambda)]} \, - \, 1
\ \geq \ 
- 3C_8 \; \alpha^{\frac{4}{49}} \; \Lambda^{-\frac{4}{49}} 
\end{align}
in \eqref{eq-VIII-vb.03}.
\end{proof}
\end{theorem}


\end{document}